\documentclass[final,12pt]{colt2023} 

\title[Fast Algorithms for a New Relaxation of Optimal Transport]{Fast Algorithms for a New Relaxation of Optimal Transport}
\usepackage{times, waingarten, framed}

\coltauthor{%
 \Name{Moses Charikar} \Email{moses@cs.stanford.edu}\\
 \addr Stanford University
 \AND
 \Name{Beidi Chen} \Email{beidic@andrew.cmu.edu}\\
 \addr Carnegie Mellow University%
 \AND
 \Name{Christopher R\'{e}} \Email{chrismre@cs.stanford.edu}\\
 \addr Stanford University
 \AND
 \Name{Erik Waingarten} \Email{ewaingar@cis.upenn.edu}\\
 \addr University of Pennsylvania
}

\begin{document}

\maketitle

\begin{abstract}%
We introduce a new class of objectives for optimal transport computations of datasets in high-dimensional Euclidean spaces. The new objectives are parametrized by $\rho \geq 1$, and provide a metric space $\calR_{\rho}(\cdot, \cdot)$ for discrete probability distributions in $\R^d$. As $\rho$ approaches $1$, the metric approaches the Earth Mover's distance, but for $\rho$ larger than (but close to) $1$, admits significantly faster algorithms. Namely, for distributions $\mu$ and $\nu$ supported on $n$ and $m$ vectors in $\R^d$ of norm at most $r$ and any $\eps > 0$, we give an algorithm which outputs an additive $\eps r$-approximation to $\calR_{\rho}(\mu, \nu)$ in time $(n+m) \cdot \poly((nm)^{(\rho-1)/\rho} \cdot 2^{\rho / (\rho-1)} / \eps)$.
\end{abstract}


\section{Introduction}

This paper is about algorithms for optimal transport problems in high dimensional Euclidean spaces. At a very high level, optimal transport problems provide a convenient metric space between probability distributions supported on vectors in geometric spaces. The most classical such problem is the Earth Mover's Distance ($\EMD$). Let $\mu$ and $\nu$ be two distributions supported on vectors in $\R^d$. The Earth Mover's Distance between $\mu$ and $\nu$, also known as the Wasserstein-$1$ distance, is given by minimizing the average distance between pairs of points sampled from a coupling $\gamma$ of $\mu$ and $\nu$:
\begin{align}
\EMD(\mu, \nu) = \min \left\{ \Ex_{(\bx, \by) \sim \gamma}\left[\| \bx - \by\|_2 \right] : \text{ $\gamma$ is a coupling of $\mu$ and $\nu$\footnotemark } \right\}. \label{eq:emd}
\end{align} 
Importantly, the Earth Mover's distance is a metric on the space of probability distributions supported on $\R^d$, which takes a ``ground metric'' (in this case, the Euclidean distances) and defines a metric over the space of distributions supported on the ground metric. \footnotetext{If the distributions $\mu$ and $\nu$ are supported on the set of points $x_1,\dots, x_n$ and $y_1, \dots, y_m$, respectively, then we can think of $\gamma$ as being specified by an $n \times m$ matrix of non-negative real numbers. The constraint ``$\gamma$ is a coupling of $\mu$ and $\nu$'' means that the $i$-th row of the $n \times m$ matrix $\gamma$ sums to $\mu_i$ and the $j$-th column of the $n\times m$ matrix sums to $\nu_j$.}
The resulting notion of similarity or dissimilarity is then used to formulate problems on approximating or learning a distribution supported in $\R^d$. 

It is no surprise that the optimal transport has become ubiquitous in machine learning. We refer the reader to the monograph~\cite{PC19} for a comprehensive overview, but a few notable examples include~\cite{KSKW15, CFTR16, ACB17}. As argued in~\cite{PC19}, the most recent progress on optimal transport for machine learning has been due to new formulations and approximation algorithms which can scale to larger problem instances. Specifically, there has been a focus on the so-called entropy-regularized optimal transport, also known as ``Sinkhorn distances,'' and (accurate) approximation algorithms which run in quadratic time (in the original representation for Euclidean inputs)~\cite{C13}. The goal of this work is to further explore such optimal transport questions from the computational perspective, where we will seek much faster \emph{sub-quadratic} algorithms for computing optimal transport distances.

\ignore{Computationally, two input distributions are specified by giving two sets of vectors $\{x_1,\dots, x_n \}$ and $\{y_1,\dots, y_m\}$ in $\R^d$ which are the supports of the two distributions, as well as two vectors $\mu \in \R^n$ and $\nu \in \R^m$, where $\mu_i$ and $\nu_j$ indicate the probability that the distributions sample $x_i$ and $y_j$, respectively. The algorithmic task is to (approximately) compute the Earth Mover's Distance $\EMD(\mu, \nu)$ between the two distributions as fast as possible. Ideally, spending as close to linear time (which is $nd$).}

\ignore{Our main regime of interest is the so-called ``high dimensional regime'', where $\omega(\log n) \leq d \leq n^{o(1)}$. In this setting, arbitrary polynomial dependencies on the underlying dimension are acceptable, since these incur a factor of $n^{o(1)}$; however, algorithms should avoid exponential dependencies on $d$, since these would result in super-polynomial factors in $n$. }

As we explain next, the algorithmic landscape for optimal transport remains very much unknown. On the one hand, the algorithms community has devoted a significant effort (\cite{C02, IT03, I04, AIK08, ABIW09, SA12, AS14,  ANOY14, BI14, AKR15, KNP19, BDIRW20, CJLW22, ACRX22}) to developing fast algorithms for approximating $\EMD$. We expand on these shortly, but, at a high level, all approaches rely on efficient spanner constructions or approximate nearest neighbor data structures. These algorithm are fast (approaching linear time), but they run into a serious approximation bottleneck. For high-dimensional Euclidean spaces, almost-linear time algorithms incur large constant-factor approximations, making these approaches undesirable.\footnote{For example, a $2$-approximation which is already oftentimes too big, incur a polynomial overhead of $(n+m)^{1/7}$~\cite{AR15}.} Instantiating these techniques for accurate, $(1\pm \eps)$-approximations degrades the algorithmic performance to essentially quadratic time.\footnote{An algorithm for $\EMD(\mu, \nu)$ (or any problem whose output is a positive real number) which achieves approximation factor $c > 1$ is an algorithm which outputs a number which is larger than $\EMD(\mu, \nu)$ and is at most $c \cdot \EMD(\mu, \nu)$ with high probability. These are \emph{multiplicative} approximations, and we will also refer to \emph{additive} $\eps r$-approximations which outputs a quantity which is up to $\pm \eps r$ from a desired quantity.}

On the other hand, algorithms for the entropy-regularized optimal transport do achieve accurate additive $\pm \eps r$ approximations for datasets of diameter $r$, but have running times which are quadratic in the original representation of the input. In particular, the input distributions are specified by the vectors in their support and the probabilities with which they are sampled. However, the first step of the algorithm involves explicitly materializing the distance matrix encoding all pairwise distances between the vectors. As the support of these distributions grows, this first step is already a major hurdle. While there have been approaches to avoid materializing the entire matrix~\cite{BRPP15, ABRN19, PC19b}, these methods consider a projection of the points onto a low-dimensional space and the resulting optimization costs (of the low-dimensional $\EMD$ or Sinkhorn distances) cannot be related back to the original distribution without a significant loss in approximation.

This work seeks to explore the best of both worlds from the algorithmic perspective. We will give a new class of objectives for optimal transport problems which also provide metric spaces for probability distributions of high-dimensional Euclidean spaces (like the Earth Mover's distance and Sinkhorn distances). The main benefit is that (i) these metrics smoothly perturb the Earth Mover's distance, (ii) admit efficient algorithms with running times which are significantly sub-quadratic (like the Earth Mover's distance), and (iii) give accurate $\pm \eps r$-approximations for distributions whose supports have diameter at most $r$ (like the Sinkhorn distances). The key will be to never explicitly compute the quadratic-size distance matrix. Instead, we show how one may implement a Sinkhorn-like update procedure using recent algorithms for the problem of kernel density estimation.

\subsection{Related Work: The Spanner Approach for $\EMD$} 

The Earth Mover's Distance can be naturally cast as an uncapacitated minimum cost flow problem. The reduction is straight-forward. One may consider the (weighted) complete bipartite graph $G = (U, V, E = U \times V, w)$ where each vertex of $U$ is a vector in the support of $\mu$ and each vertex in $V$ is a vector in the support of $\nu$ and the weights (or cost) $w$ of an edge $e = (i, j)$ is $w(e) = \| x_i - y_j\|_2$. The distributions may then be written as vectors $\mu \in \R^n$ and $\nu \in \R^m$ which encode the ``supply'' and ``demand'', and the Earth Mover's Distance is the minimum cost flow on $G$ according to the supply/demands $\mu$ and $\nu$ with costs $w$ (there is no need for capacities in this reduction). Over the years, graph algorithms have become incredibly efficient, so applying graph-based min-cost flow solvers with the reduction above gives exact algorithms for $\EMD$ running in time $(nm)^{1+o(1)}$.\footnote{The relevant citation for a fast min-cost flow algorithm is the recent breakthrough of~\cite{CKLPPS22}. These give exact  algorithms for graphs whose time in almost-linear in the number of edges, $nm$ of the graph. The other relevant citation is \cite{S17}, giving algorithms for $1+\eps$-approximation to uncapacitated min-cost flow in the same amount of time, which suffices for $\EMD$.}

The above approach paves the way for faster approximation algorithms by using graph spanners. For any $c > 1$, one seeks a graph $H$ with substantially fewer edges on the vertex set $U$ and $V$. The desired property is that for any $i \in U$ and $j \in V$, the total length of the shortest path between $i \in U$ and $j \in V$ along edges of $H$ should be a factor of $c$-approximation to the distance between the underlying vectors $x_i$ and $y_j$. Running the min-cost flow algorithms on $H$ is faster (since there are fewer edges), and give a $c$-approximation for $\EMD$.  While sparse spanners for Euclidean distances do exist, as the approximation $c$ approaches $1+\eps$, the size of these spanners become $mn$. 

Instead, the focus has been on obtaining sparse spanners for (large) constant factor approximations. For example, for any $c > 1$,~\cite{HIS13} gives $c$-spanners of size $(n+m)^{1+1/c^2}$ for Euclidean spaces ($\ell_2$) in time $\tilde{O}((n+m)^{1+1/c^2})$ (which is fast when we allow a large $c$). The other approach, taken in \cite{AS14}, does not explicitly use a spanner, but uses an approximate nearest neighbor search data structure. The resulting time and approximation depends on the time and approximation for nearest neighbor search, but similarly to before, the approximation is large when the algorithms are fast.

\subsection{Related Work: Sinkhorn Distances}

The algorithm which is widely used for computing an optimal transport is the Sinkhorn algorithm for entropy-regularized optimal transport~\cite{C13, AWR17} (see also, the recent work~\cite{PLHPB20, LNNPHH21}). Given two distributions $\mu$ and $\nu$ supported on vectors in $\R^d$, the entropy-regularized optimal transport introduces an entropic regularization term to the the Earth Mover's distance. Specifically, for any $\eta \geq 0$, it optimizes
\begin{align*}
\SNK_{\eta}(\mu, \nu) = \min\left\{ \Ex_{(\bx, \by) \sim \gamma}\left[ \| \bx - \by\|_2 \right] - \eta H(\gamma) : \text{ $\gamma$ is a coupling of $\mu$ and $\nu$ } \right\}.
\end{align*}
The main benefit is that the algorithm for optimizing $\SNK_{\eta}(\mu, \nu)$ performs extremely well. The algorithm used is iterative, and uses $\poly(1/(\eta\eps))$ iterations to output a solution which is an additive $\pm \eps r$-approximation (where $r$ is the maximum distance between any pair of points in the support of $\mu$ and $\nu$). Oftentimes, the maximum distance $r$ is not too large (for example, it is at most $2$ on the unit sphere), making the algorithm very desirable in practice. However, the main downside is that the algorithm explicitly computes the $nm$-distance matrix of pairwise distances of vectors in the support of $\mu$ and $\nu$. Indeed, the algorithm does not use the fact that distances are Euclidean and generalizes to non-Euclidean metrics. The main downside is that, for distributions on Euclidean spaces, the description of the input (of size $O(d(n+m))$) is blown up to a quadratic $nm$-size distance matrix, which can be a major bottleneck in the computation if $n$ and $m$ are very large. Finally, it is important to note that, we currently do not know whether the original Earth Mover's distance admits a similar $\eps r$-approximation for bounded datasets in time substantially smaller than $nm$. 

\subsection{Our Contributions} 

This paper addresses the following questions:
\begin{enumerate}
\item\label{q:one} Do there exists optimal transport metrics which do admit good approximations in significantly sub-quadratic time? In particular, can we match the approximation guarantees from Sinkhorn distances with the algorithmic techniques from the Earth Mover's distance?
\item\label{q:two} Can one combine techniques, like locality-sensitive hashing (LSH) and embeddings, with the alternating updates procedure in Sinkhorn's algorithm even though approximations incurred from using LSH and embeddings tend to incur large constant factors?
\end{enumerate}
Our main contribution is introducing a class of objective functions for optimal transport computations. The new objectives $\calR_{\rho}(\mu, \nu)$ are parametrized by $\rho \geq 1$ and provide metric spaces over discrete distributions in $\R^d$. As $\rho$ approaches 1, $\calR_{\rho}(\mu, \nu)$ approaches $\EMD(\mu, \nu)$, but enjoys favorable computational properties. In particular, we will show that $\calR_{\rho}(\mu, \nu)$ may be approximated up to additive $\eps r$-error for datasets of diameter at most $r$ in time which is near-linear (for small $\rho$ close to $1$). We view $\rho$ as introducing a new ``knob'' for the Earth Mover's distance: as $\rho \to 1$, the metrics $\calR_{\rho}(\cdot,\cdot)$ approach $\EMD(\cdot,\cdot)$; however, for $\rho$ close to (but not too close to) 1, very fast algorithms with accurate approximations are possible. Thus, our new algorithm gives a positive answer to Question~\ref{q:one}. Namely, if one is willing to change the problem slightly, one can achieve the approximation guarantees of Sinkhorn distances with the running times like the Earth Mover's distance.

While Question~\ref{q:two} is inherently vague, such techniques are known in a related algorithmic context. One of our main conceptual contributions is drawing a connection to \emph{kernel density estimation} \cite{CS17, BCIS18, SRBCL19, CKNS20, BIMW21, BIKSZ22}. 
The algorithms developed in that context use locality-sensitive hashing and embeddings, but are still able to output $(1\pm \eps)$-approximations. In particular, a key feature of those works is that the distortion incurred by locality-sensitive hashing and embeddings factors into the running time of the algorithm and not the final approximation. In summary, our main conceptual contributions may be summarized as follows:
\ignore{From a technical perspective, the novel aspect is focusing attention on a different formulation which will imply the new results for $\EMD$. We will define a version of the Earth Mover's distance which optimizes an $\ell_{\rho}$-regression version of the problem (see Definition~\ref{def:ell-rho-regression}). The case $\rho = 1$ will correspond exactly to $\EMD$, but setting $\rho$ close to, but greater than, $1$ will give a (relatively) minor perturbation of $\EMD$, and at the same time be algorithmically much simpler. As far as we can tell, the $\ell_{\rho}$-regression version of $\EMD$ has not been studied before, and the algorithms we provide (which will achieve a type of $(1\pm \eps)$-approximations) novel, even when specialized to $\ell_1$ and $\ell_2$. Our main conceptual contribution can be summarized as follows:}
\begin{itemize}
\item There exists a class of optimal transport metrics parametrized by $\rho$ which smoothly perturb the Earth Mover's distance (approaching $\EMD$ as $\rho \to 1$).
\item For a small setting of $\rho > 1$, these problems can be optimized in significantly sub-quadratic time to arbitrarily accurate additive approximations for bounded datasets.
\end{itemize}
We believe the new problem formulation and the ideas behind the algorithm will lead to improvements in practical algorithms for optimal transport metrics. We emphasize that there are no algorithmic approaches that achieve $(1\pm \eps)$-approximations or $\eps r$-additive approximations for either $\EMD$ nor $\SNK$ in time $n^{1.99}$. In addition, there is some reason to believe that this may be impossible for $\EMD$~\cite{R19}. By changing the problem and allowing a small additive error, we avoid the large constant factors. \ignore{In cases where one expects $\EMD$, $\SNK$ and $\calR_{\rho}$ to be similar, the algorithms may be especially useful (for instance, when optimal couplings $\gamma$ already have costs tend to be ``spread'' ). \red{Moses: do we know some class of instances where this happens?}} We also suggest looking at Section 4 of~\cite{BDIRW20}, who group algorithms by their running times; the new techniques achieve the accurate approximations of the ``quadratic time'' algorithms, even though they run much faster (at least in theory). \ignore{\red{Moses: The paper actually doesn't use the word ``regimes''. Instead, should we say: We also suggest looking at Section 4 of~\cite{BDIRW20}, who group algorithms by their running time; the new techniques achieve the accurate approximations of the ``quadratic time'' algorithms, even though they could (at least in theory) run much faster.}}

\paragraph{Outline.} The next section gives the new objective $\calR_{\rho}(\mu, \nu)$ and states our main Theorem~\ref{thm:main}. We will overview the components of the proof in the next section. Then, we give a description of the main algorithm while assuming algorithms for estimating the gradients and the penalty term.

\newcommand{\rmin}{r_{\min}}

\section{The Definition of $\ell_{\rho}$-Optimal Transports}

For any dimension $d \in \N$, let $\mu$ and $\nu$ denote two discrete distributions supported on $n$ and $m$ point masses in $\R^d$, respectively. More specifically, $\mu$ is specified by $n$ points $x_1,\dots,x_n \in \R^d$ and corresponding weights $\mu_1,\dots, \mu_n \in \R_{> 0}$ where $\sum_{i=1}^n \mu_i = 1$, and $\nu$ is specified by $m$ points $y_1,\dots, y_m \in \R^d$ with the corresponding weights $\nu_1,\dots, \nu_m \in \R_{> 0}$ with $\sum_{i=1}^m \nu_i = 1$ (note that we can always assume that $\mu_i$ and $\nu_j$ are strictly positive by a linear-time scan which can remove points of weight-$0$). One ought to think of $d = \omega(\log n)$, so we seek algorithms which overcome the ``curse of dimensionality'' and do not have running times which scale exponentially in $d$. 

For any parameter $\rho > 1$, we seek to optimize the following objective, which will specify a metric space over probability distributions which relax the optimal transport problem (Lemma~\ref{lem:triangle} in Appendix~\ref{sec:basic-prop}):
\begin{align}
\calR_{\rho}(\mu, \nu) &= \min \left\{ \left( \Ex_{\substack{\bi \sim \mu \\ \bj \sim \nu}}\left[\left( \frac{\gamma_{\bi\bj}}{\mu_{\bi} \nu_{\bj}} \cdot \|x_{\bi}- y_{\bj}\|_2\right)^{\rho}\right] \right)^{1/\rho}: \gamma \text{ is a coupling of $\mu$ and $\nu$} \right\}. \label{eq:objective}
\end{align}
In words, for any coupling $\gamma$ between the distributions $\mu$ and $\nu$, one may associate an $n m$-dimensional vector encoding the costs associated with each point-mass. Each point $x_i$ from $\mu$ and $y_j$ from $\nu$, the coupling $\gamma$ transports $\gamma_{ij}$ ``mass'' from $x_i$ to $y_j$ and pays a function of the distance between $x_i$ and $y_i$ times $\gamma_{ij} / (\mu_i \nu_j)$. In $\calR_{\rho}(\mu, \nu)$, we optimize the normalized $\ell_{\rho}$-norm of the cost vector. Notice that, when $\rho = 1$, $\calR_{\rho}(\mu, \nu)$ is the Earth Mover's distance distance between $\mu$ and $\nu$. As we vary $\rho \geq 1$, one may relate the $\ell_{\rho}$- and $\ell_1$-norm, implying
\begin{align*}
\EMD(\mu, \nu) \leq  \calR_{\rho}(\mu, \nu) \leq \sup_{i, j} \left|\frac{1}{\mu_i \nu_j}\right|^{(\rho-1)/\rho} \EMD(\mu, \nu).
\end{align*}
When $\rho > 1$, we will obtain a sequence of (as we will see) computationally easier metric spaces which approach $\EMD(\mu, \nu)$. The key is that performing this modification will allow for significantly faster algorithms in terms of $n$ and $m$ (the number of points), while having a dependence on $\rho$ which will be $2^{O(\rho / (\rho - 1))}$. 

We view $\rho > 1$ as a desired computational ``knob,'' which allows one to tradeoff the running time of an algorithm and the metric's relation to $\EMD$. Note that, in a $c$-approximation algorithm for $\EMD$, $c$ also trades-off faster/slower running times for looser/tighter relations to $\EMD$. The difference, however, is that for any $\rho > 1$, $\calR_{\rho}(\cdot,\cdot)$ is still a metric space over probability distributions (and the same cannot be said of a $3$-approximation to $\EMD$). The specific choice of metric space ($\EMD$, Wasserstein-$p$, or $\SNK_{\eta}$) is oftentimes flexible, so long as it captures the desired notion of similarity/dissimilarity of distributions. The hope is that for moderate values of $\rho$, $\calR_{\rho}(\mu, \nu)$ suffices for downstream applications, and captures the desired properties of an optimal-transport $\gamma$.

From a more technical perspective, (\ref{eq:objective}) encourages couplings $\gamma$ whose contribution to the cost vector is ``spread'', so that the $\ell_{\rho}$-norm will be small. The main advantage is that, using a connection to recent work on kernel density estimation in high-dimensions \cite{BCIS18} and scaling approaches to entropy regularized optimal transport \cite{C13,AWR17}, we give very efficient (and simple) algorithms for approximating $\calR_{\rho}(\mu, \nu)$.

\paragraph{Notation for Running Time Bounds.}\label{sec:notation} We will use the following notation in order to describe the running time bounds. The focus is on improving on the dependence on $n$ and $m$ when estimating optimal transports, so we use the notation $\poly^*(f)$ to denote a fixed polynomial function of $f$, and which hides poly-logarithmic factors $n, m$, $\delta$ (the failure probability), $\eps$ (the accuracy) and $r$ (the radius of the dataset). In addition, since we will incur a polynomial dependence on $\eps$, we will automatically apply the Johnson-Lindenstrauss lemma and assume that $d = O(\log(nm) / \eps^2)$. 

\begin{theorem}\label{thm:main}
There exists a randomized algorithm with the following guarantees. The algorithm receives as input
\begin{itemize}
\item Two sets of points $\{x_1,\dots, x_n \}$ and $\{ y_1,\dots, y_m \}$ in $\R^d$ where the maximum pairwise distance between points $\sup_{i, j} \|x_i - y_j\|_2 \leq r$.
\item Two vectors $\mu \in \R^n_{\geq 0}$ and $\nu \in \R^m_{\geq 0}$ whose coordinates sum to $1$ and encode the distributions over $\{x_1,\dots, x_n \}$ and $\{y_1,\dots, y_n\}$, respectively.
\item An accuracy parameter $\eps > 0$, a failure probability $\delta > 0$, and a parameter $\rho \in [1,2]$.
\end{itemize}
The algorithm runs in time $(n + m) \cdot \poly^*((nm)^{(\rho-1)/\rho} \cdot 2^{\rho/(\rho-1)} / \eps)$, and outputs an estimate $\hat{\boldeta} > 0$ which satisfies
\begin{align*}
\left| \hat{\boldeta} - \calR_{\rho}(\mu, \nu) \right| \leq \eps \cdot r
\end{align*}  
with probability at least $1-\delta$.
\end{theorem}

The main advantage of Theorem~\ref{thm:main} is that it does not pay the quadratic $nm$-factor in the running time and at the same time obtains accurate approximations. In particular, suppose we consider a setting of $\rho$ which is $\rho = 1 + 1/\sqrt{\log(nm)}$, then the corresponding running time of Theorem~\ref{thm:main} to approximate $\calR_{\rho}(\mu, \nu)$ up to an additive $\pm \eps r$ becomes
\[ (n+m)^{1+o(1)} \cdot \poly(1/\eps). \]
Generally, as $\rho$ becomes close to $1$, the metric $\calR_{\rho}(\cdot,\cdot)$ approaches $\EMD(\cdot,\cdot)$ and the dependence on $n$ and $m$ becomes better, since $(n+m) \cdot (nm)^{O((\rho-1)/\rho)}$. However, one does not want to set $\rho$ to be too close to $1$, since the factor of $2^{O(\rho / (\rho-1))}$ may begin to dominate. 

\begin{remark}[Challenges when $\rho \to 1$]
In order to use $\calR_{\rho}(\cdot, \cdot)$ to approximate $\EMD(\cdot,\cdot)$ up to $(1+\eps)$-factor, one would need to set $\rho$ to roughly $1 + O(\eps / \log(nm))$; however, this approach runs into a technical challenge. There is a concrete sense in which the parameter $\rho \geq 1$ adds a certain ``smoothness'' which is not present in $\EMD$. At a very high level, we show that an additive approximation of $\calR_{\rho}$ reduces to queries for ``smooth'' kernel density evaluation~\cite{BCIS18} which suffer an exponential dependence on $\rho / (\rho - 1)$. With $\rho = 1 + O(\eps / \log(nm))$, this dependence would become $(nm)^{O(1/\eps)}$---worse than the $(nm)^{1+o(1)}$ time required from prior work.
\end{remark}

\subsection{Proof of Theorem~\ref{thm:main} Overview}

We overview the major components of the proof of Theorem~\ref{thm:main}. While (relatively minor) technical challenges arise when fleshing out the details, the structure and algorithm proceed with the following plan.

\paragraph{The Duals of $\EMD(\mu,\nu)$ and $\calR_{\rho}(\mu, \nu)^{\rho}$.} The challenge in optimizing $\calR_{\rho}(\mu,\nu)^{\rho}$ (which also appears in $\EMD(\mu, \nu)$) is that an algorithm cannot even write down the explicit description of the optimization, nor can it explicitly maintain a coupling $\gamma$, since this requires $\Omega(nm)$ values. On the other hand, both $\EMD(\mu, \nu)$ and $\calR_{\rho}(\mu,\nu)^{\rho}$ only have $n+m$ equality constraints, so the duals are maximization problems over $n+m$ variables (one for each constraint). The approach will be to show that, using data structures for kernel density estimation, we can implicitly maximize the dual of $\calR_{\rho}(\mu, \nu)^{\rho}$ while only maintaining the $n+m$ dual variables.

To see the connection, we first write down the dual for $\EMD(\mu,\nu)$, which has $n+m$ variables $\alpha_1,\dots, \alpha_n$ and $\beta_1,\dots, \beta_m$ and asks to maximize
\begin{align}\label{eq:emd-dual}
\EMD(\mu,\nu)= \max_{\substack{\alpha \in \R^n \\ \beta \in \R^m}}\left\{ \sum_{i=1}^n \mu_i \alpha_i - \sum_{j=1}^m \nu_j \beta_j  :    \forall (i, j) \in [n] \times [m], \alpha_i - \beta_j \leq \| x_i - y_j \|_2\right\}.
\end{align}
For $\rho > 1$, the H\"{o}lder conjugate $s >1$, is the number satisfying $1/\rho + 1/s = 1$. The dual for $\calR_{\rho}(\mu, \nu)^{\rho}$ is the following unconstrained maximization problem on $n+m$ variables $\alpha_1,\dots, \alpha_n$ and $\beta_1,\dots, \beta_m$,
\begin{align}\label{eq:ell-rho-dual}
\calR_{\rho}(\mu, \nu)^{\rho} = \max_{\substack{\alpha \in \R^n \\ \beta \in \R^m}}\left\{ \sum_{i=1}^n \mu_i \alpha_i - \sum_{j=1}^m \nu_j \beta_j - \frac{1}{s} \left( 1 - \frac{1}{s}\right)^{s-1} \sum_{i=1}^n \sum_{j=1}^m \mu_i \nu_j \left( \dfrac{(\alpha_i-\beta_j)^+}{\|x_i - y_j\|_2}\right)^{s} \right\},
\end{align}
where we consider $\frac{0}{0} = 0$, {and $(\alpha_i - \beta_j)^+$ is $\alpha_i - \beta_j$ if positive and $0$ otherwise}. Note the difference: in (\ref{eq:emd-dual}), there are $nm$ hard constraints which enforce $(\alpha_i - \beta_j)^+ / \| x_i - y_j \|_2 \leq 1$ for every $i \neq j$. In (\ref{eq:ell-rho-dual}), the $nm$ constraints are relaxed. The optimization is allowed to set $\alpha_i - \beta_j$ larger than $\|x_i - y_j \|_2$, but pays a penalty in the objective proportional to $((\alpha_i - \beta_j)^+ / \|x_i - y_j \|_2)^{s}$. As $\rho$ gets closer to $1$, the H\"{o}lder conjugate $s$ becomes larger, and the penalty becomes more pronounced. For simplicity in the notation, we will write
\[ g(\alpha, \beta) \eqdef \sum_{i=1}^n \mu_i \alpha_i - \sum_{j=1}^m \nu_j \beta_j - \frac{1}{s} \left(1 - \frac{1}{s} \right)^{s-1} \sum_{i=1}^n \sum_{j=1}^m \mu_i \nu_j \left(\frac{(\alpha_i - \beta_j)^+}{\|x_i - y_j\|_2} \right)^s. \]

\paragraph{Partial Derivatives via Kernel Density Estimation} Since (\ref{eq:ell-rho-dual}) is a concave maximization problem, a simple approach is to simulate a gradient ascent algorithm on the dual variables $\alpha_1,\dots, \alpha_n$ and $\beta_1,\dots, \beta_m$, where we update in the direction of the partial derivatives. The partial derivatives with respect to $\alpha_i$ and $\beta_j$ are given by
\begin{align}
\frac{\partial g}{\partial \alpha_i} &= \mu_i\left(1  - \left(1 - \frac{1}{s} \right)^{s-1} \sum_{j=1}^m \nu_j \cdot \dfrac{((\alpha_i - \beta_j)^+)^{s-1}}{\|x_i - y_j \|_2^s} \right) \label{eq:partial-alpha}\\
\frac{\partial g}{\partial \beta_j} &= -\nu_j\left(1 - \left(1 - \frac{1}{s} \right)^{s-1} \sum_{i=1}^n \mu_i\cdot \dfrac{((\alpha_i - \beta_j)^+)^{s-1}}{\|x_i - y_j \|_2^s} \right). \label{eq:partial-beta}
\end{align}
Importantly, the partial derivatives depend on $\mu_i$ and $\nu_j$ and a weighted sum of $1/\|x_i - y_j\|_2^s$. First, note that we receive $\mu$ and $\nu$ as input, so $\mu_i$ and $\nu_j$ are $n+m$ constants throughout the execution. The weighted sums are the more challenging parts, and for these we use the kernel density estimation data structures. We interpret $\sfK(x_i, y_j) = 1 / \|x_i - y_j \|_2^s$ as a ``smooth'' kernel, similar to the Student-$t$ Kernel studied in \cite{BCIS18}. These smooth kernels decay polynomially as a function of the distance $\| \cdot\|_2$ and admit very efficient data structures. Specializing the results of \cite{BCIS18} for $\sfK$, they give data structures which preprocess a set of points $P$ and can support $(1\pm \eps)$-approximate kernel evaluation queries of the form $\sum_{x \in P} \sfK(x, y)$ for any $y \in \R^d$. The query complexity is $\poly^*(2^s / \eps)$ and $s$ becomes $\rho / (\rho - 1)$. In order to use these for Theorem~\ref{thm:main}, we incorporate the weights $((\alpha_i - \beta_j)^+)^{s-1}$ by augmenting those data structures in Section~\ref{sec:data-structures} (we overview the augmentations shortly). Once this is done, the algorithm can initialize $\alpha \in \R^n$ and $\beta \in \R^m$ to $0$ and effectively update $\alpha$ and $\beta$ in the directions of the partial derivatives in order to increase the objective function.

The only remaining challenge is setting the step size of the update, and ensuring that the function is smooth enough. Note that because of the non-linear penalty term, there is no global Lipschitz constant, but we will argue that our optimization always remains within a smooth enough region if the step size is set appropriately. We do this final argument by applying a simple preprocessing step. The preprocessing will guarantee that the distance between any $x_i$ and $y_j$ is always between $\eps r$ and $r$ (which changes $\calR_{\rho}(\mu, \nu)$ by at most $\eps r$), and that every non-zero element of the support of $\mu$ and $\nu$ is sampled with at least some probability. This means that an update which changes some $\alpha$ or $\beta$ does not change the penalty term significantly (because the fact that the distance $\|x_i -y_j\|_2$ in the denominator is at least $\eps r$ ensures the penalty does not blow up).

\paragraph{Augmenting Kernel Density Estimates to Incorporate Weights} For $s > 1$, we want to maintain a set of points $P = \{ x_1,\dots, x_{n} \}$ in $\R^d$, where each point is associated with a weight $\alpha_1,\dots, \alpha_n \in \R$ and a parameter $\mu_1,\dots, \mu_i$ which are between $1/\poly(n)$ and $1$. A query is specified by another vector $y \in \R^d$ and its weight $\beta$, and the task is to output
\begin{align}
\sum_{i=1}^{n} \mu_i \cdot ((\alpha_i - \beta)^+)^{s-1} \cdot \sfK(x_i, y), \label{eq:kde-query}
\end{align}
where $\sfK(x_i, y) = 1 / \|x_i - y\|_2^{s}$. We will augment the data structures from \cite{BCIS18} as follows. First, partition $P$ into $O(\log n /\eps)$ ranges which partition $[1/\poly(n), 1]$ according to powers of $1+\eps$ so as to assume that $\mu_j$ is the same within each range. Note that we know the weights $\mu_1,\dots, \mu_n$ during the preprocessing, so that we may perform this partition; however, since we do not know $\beta$ during the preprocessing, we cannot similarly partition according to the value of $(\alpha_i - \beta)^s$. 

Instead, we will proceed with the following. For each range $j$, the resulting set $P_j$ is stored sorted in a binary tree according to the weights $\alpha$, and let $\alpha_{\max}$ be the largest weight. Each internal node holds a data structure of~\cite{BCIS18} maintaining points in its subtree. When a query $(y, \beta) \in \R^d \times \R$ comes, one may perform the following:
\begin{enumerate}
\item Let $\bxi$ be uniformly drawn from the interval $[0, (\alpha_{\max} - \beta)^{s-1}]$.
\item Find the value $\beta + \bxi^{1/(s-1)}$ in the binary tree, and we consider the $k = O(\log n)$ nodes which partition the interval $[\beta + \bxi^{1/(s-1)}, \alpha_{\max}]$.
\item Query all $k$ kernel evaluation data structures stored at those nodes. If $\hat{\boldeta}_1,\dots, \hat{\boldeta}_k$ are the estimates output by the $k$ data structures with $y$, output $ (\alpha_{\max} - \beta)^{s-1} \sum_{\ell=1}^k \hat{\boldeta}_{\ell}$.
\end{enumerate}
The main observation is that the sampling automatically incorporates weights. For example, suppose the data structures of \cite{BCIS18} were exact, then our estimate is an unbiased estimator of (\ref{eq:kde-query}):
\begin{align*}
\Ex_{\bxi}\left[ (\alpha_{\max} - \beta)^{s-1} \sum_{\ell=1}^k \hat{\boldeta}_{\ell} \right] &= \sum_{i=1}^{m} (\alpha_{\max} - \beta)^{s-1} \cdot \Prx_{\bxi}\left[ \alpha_i \geq \beta + \bxi^{1/(s-1)} \right] \cdot \sfK(x_i, y),
\end{align*}
and the probability that $\alpha_i \geq \beta + \bxi^{1/(s-1)}$ is exactly $((\alpha_i - \beta)^+)^{s-1} / (\alpha_{\max} - \beta)^{s-1}$. The variance of the above estimation is too large (which occurs because $\alpha_{\max} \gg \beta$), so we make the following minor modification. We partition the interval $[\beta, \alpha_{\max}]$ into poly-logarithmic, geometrically increasing groups, and perform the above process for each group. This is then enough to bound the variance.

\section{A Gradient Ascent Algorithm}\label{sec:sgd}

\newcommand{\err}{\mathrm{err}}
\newcommand{\EstAlpha}{\texttt{Est-Alpha}}
\newcommand{\EstBeta}{\texttt{Est-Beta}}
\newcommand{\EstPenalty}{\texttt{Est-Penalty}}
\newcommand{\bomega}{\boldsymbol{\omega}}

\subsection{A Simple Preprocessing}\label{sec:preprocess}

Before we give the description of the algorithm, we will run a simple preprocessing step which simplifies our input. We will think of $\mu$ and $\nu$ as the distribution over $\{x_1,\dots, x_n \}$ and $\{y_1,\dots, y_m\}$, respectively. For small parameters $\sigma, \sigma_{\mu}, \sigma_{\nu} > 0$, we will define the distributions $\mu'$ and $\nu'$ in the following way:
\begin{itemize}
\item First, we consider the points $x_1',\dots, x_n'$ and $y_1', \dots, y_m'$ in $\R^{d+1}$ where we append a coordinate and we let $x_i' = (x_i, \sigma r)$ and $y_j' = (y_j, 0)$. This way, we guarantee that for every $i \in [n]$ and $j \in [m]$, we satisfy $\sigma r \leq \| x_i' - y_j' \|_2 \leq r \sqrt{1 + \sigma^2}$ (where the upper bound follows from the fact $\|x_i - y_j\|_2 \leq r$).
\item We define the sets $L_{\mu} \subset [n]$ and $L_{\nu} \subset [m]$ for the indices of $\mu$ and $\nu$ which have low probability, i.e., $L_{\mu} = \{ i \in [n] : \mu_i < \sigma_{\mu} / n \}$ and $L_{\nu} = \{ j \in [m] : \nu_j < \sigma_{\nu} / m\}$. We denote $\zeta_{\mu} = \sum_{i \in L_{\mu}} \mu_i \leq \sigma_{\mu}$  and $\zeta_{\nu} = \sum_{j \in L_{\nu}} \nu_j \leq \sigma_{\nu}$. The distribution $\mu'$ is supported on the points $x_1',\dots, x_n'$, and $\nu'$ is supported on the points $y_1',\dots, y_n'$ given by
\[ \mu_i' = \left\{ \begin{array}{cc} 0 & i \in L_{\mu} \\ 
						\mu_i / (1-\zeta_{\mu}) & i \in [n] \setminus L_{\mu} \end{array} \right. \qquad\text{and}\qquad \nu_j' = \left\{ \begin{array}{cc} 0 & j \in L_{\nu} \\
											   \nu_j / (1-\zeta_{\nu}) & j \in [m] \setminus L_{\nu}\end{array} \right. .  \]
\end{itemize}

The above transformations has the benefit that we now have a lower bound on the minimum distance between any point from the support of $\mu'$ and any point from the support of $\nu'$, while only increasing the maximum distance by at most a factor of $\sqrt{1+\sigma^2}$. Furthermore, the distributions $\mu'$ and $\nu'$ have all elements of their support with probability at least $\sigma_{\mu} / n$ and $\sigma_{\nu} / m$, respectively, since we have removed the low-probability items. Thus, the algorithm below will apply the above perturbation, and we may assume throughout the execution the corresponding properties of $\mu$ and $\nu$. Note that as long as we ensure the parameter 
\[ \left( n^{\rho-1} \cdot \frac{\sigma}{\sigma_{\mu}^{\rho-1}} + \sigma_{\mu} \right)^{1/\rho} \leq \eps \qquad\text{and}\qquad \sigma_{\nu}^{1/\rho} \leq \eps,\]
then by the triangle inequality, we will have $\calR_{\rho}(\mu', \nu')$ is up to an additive $2\eps r$, the same as $\calR_{\rho}(\mu, \nu)$. In particular, we can let $\sigma_{\nu}$ be $\eps^{\rho}$ and $\sigma_{\mu} = \eps^{\rho} / n$ and $\sigma = \eps^{\rho}$.

\subsection{Description of the Algorithm}\label{sec:alg}

We will assume hence-forth that our input distributions $\mu$ and $\nu$, whose support is $\{ x_1,\dots, x_n \}$ and $\{y_1,\dots, y_n\}$ satisfy:
\begin{itemize}
\item Every $i \in [n]$ and $j \in [m]$, the distance $\|x_i - y_j\|_2$ is always between $\sigma r$ and $r$ (for a small parameter $\sigma > 0$, we have $r \sqrt{1 + \sigma^2} \leq 2r$ so, in order to simplify the notation, one may think of $\sigma$ as being decreased by a factor of $2$).
\item The distributions $\mu$ and $\nu$ have a ``granularity'' property, so that every $i \in [n]$ for which $\mu_i$ is non-zero is at least $\sigma_{\mu} / n$, and every $j \in [m]$ for which $\nu_j$ is non-zero is at least $\sigma_{\nu} / m$. This will allow us to upper bound $1/(\mu_i \nu_j) \leq mn / (\sigma_{\mu} \sigma_{\nu})$.
\end{itemize}

The algorithm will maintain a setting of the dual variables $(\alpha_t, \beta_t) \in \R^{n+m}$ which it will update in each iteration, and it will seek to maximize
\[ g(\alpha, \beta) = \sum_{i=1}^n \mu_i \alpha_i - \sum_{j=1}^m \nu_j \beta_j - C_s \sum_{i=1}^n \sum_{j=1}^m \mu_i \nu_j \left(\frac{(\alpha_i -\beta_j)^+}{\|x_i - y_j\|_2} \right)^{s}. \]
In the description of the algorithm below, we will assume access to three sub-routines $\EstAlpha$, $\EstBeta$, and $\EstPenalty$ which we specify later (see Subsection~\ref{sec:sub-routine-guarantees} for a description of the guarantees). At a high level, the sub-routines $\EstAlpha$ will help us get an approximation of the gradient $\nabla g(\alpha_t, \beta_t)$ along directions in $\alpha$, and $\EstBeta$ will help us get an approximation of the gradient $\nabla g(\alpha_t, \beta_t)$ along directions in $\beta$. The sub-routine $\EstPenalty$ will come in at the end, since we will need to estimate the ``penalty'' term in order to output an approximation to $g(\alpha_t, \beta_t)$. We will instantiate the algorithm with the following parameters:
\begin{itemize}
\item \textbf{Accuracy of Terminating Condition}: we denote this parameter $\eps_2 > 0$, which will be set to $c_0 \cdot \eps \cdot \left(\sigma_{\mu} \sigma_{\nu} / (mn)\right)^{(\rho-1)/\rho}$ for a small enough constant $c_0 > 0$. This parameter will dictate when our algorithm has found a dual solution which is close enough to the optimal one.
\item \textbf{Accuracy for Estimation}: There are two parameters which specify the accuracy needed in the estimations $\EstAlpha$ and $\EstBeta$. We let $\eps_1 > 0$ denote the multiplicative error bound which we will tolerate, set to $c_1 \eps_2 / s$ for a small enough constant $c_1$, and $\tau$ which will be an additive error bound will may be interpreted as a granularity condition on the weights $\alpha, \beta$. It will suffice to set $\tau = c_2 \eps_2$, but the final dependence on $\tau$ will be poly-logarithmic in $1/\tau$, the notation $\poly^*(\cdot)$ will suppress it.
\item \textbf{Step Size of Gradient Ascent}: The parameter $\lambda \geq 0$ will denote the step size of our gradient ascent algorithm. We set $\lambda = c_3 \eps_2 \cdot (\sigma / s)^2 \cdot r^{\rho}$, for a small constant $c_3 > 0$. 
\end{itemize}
We will also consider a small enough parameter $\delta > 0$ which will denote the failure probabilities of our estimation algorithms. The final dependence on $\delta$ is only poly-logarithmic, so it will suffice to set $\delta$ to be a small enough polynomial factor of all parameters of the algorithm (i.e., $n,m,1/\eps, 1/\sigma, s$) such that all executions of $\EstAlpha$, $\EstBeta$, and $\EstPenalty$ succeed with high probability. For simplicity in the notation, we will drop $\delta$ from the notation, and assume that all executions of our (randomized) sub-routines succeed. 

\begin{figure}
	\begin{framed}
		\noindent Main Algorithm for Computing $\calR_{\rho}(\mu, \nu)^{\rho}$ (after preprocessing from Subsection~\ref{sec:preprocess}).
		\begin{flushleft}
			\noindent {\bf Input:}  Two vectors $\mu \in \R^n$ and $\nu \in \R^m$ which encode two distributions supported on the points $\{x_1,\dots, x_n\}$ and $\{y_1,\dots, y_m \} $ in $\R^d$, and an accuracy parameter $\eps > 0$. \\
			\noindent {\bf Assumptions:} Every $i \in [n]$ and $j \in [m]$ satisfies $\sigma r \leq \|x_i - y_j \|_2 \leq r$. Every $i \in [n]$ has $\mu_i \geq \sigma_{\mu} / n$ and every $j \in [m]$ has $\nu_j \geq \sigma_{\nu} / m$. We refer to parameters $\eps_1, \eps_2, \tau$ and $\lambda$ specified above (as a function of $\eps$), and access to sub-routines $\EstAlpha$, $\EstBeta$, and $\EstPenalty$.\\ \vspace{0.5cm}
			\noindent We initialize $(\alpha_0, \beta_0) = (0, 0) \in \R^{n +m}$ and iteratively perform the following updates for $t \geq 1$:
		\begin{itemize}
\item \textbf{Run Estimates}: We execute $\EstAlpha(\alpha_t, \beta_t, \eps_1, \tau)$ which produces as output a sequence of $n$ numbers $\boldeta_1,\dots, \boldeta_n$, and we execute $\EstBeta(\alpha_t, \beta_t, \eps_1, \tau)$ which returns a sequence of $m$ numbers $\bxi_1,\dots, \bxi_m$.
\item \textbf{Update $\alpha$'s}: If $\sum_{i=1}^n \mu_i |1 -  \boldeta_i| \geq \eps_2$, we will update the $\alpha$'s by letting 
\[ \alpha_{t+1} = \alpha_t + \lambda \cdot \sign(\ind - \boldeta), \]
where  $\sign(\ind - \boldeta)$ is the vector in $\{-1,1\}^n$ where the $i$-th entry is $\sign(1 - \boldeta_i)$. We also update $\beta_{t+1} = \beta_t$ and increment $t$, beginning a new iteration.
\item \textbf{Update $\beta$'s}: If $\sum_{j=1}^m \nu_j |\bxi_i - 1| \geq \eps_2$, then we will update the $\beta$'s by letting
\[ \beta_{t+1} = \beta_t - \lambda \cdot \sign(\bxi - \ind).\] 
We update $\alpha_{t+1} = \alpha_t$ and increment $t$, beginning a new iteration.
\item \textbf{Termination}: Otherwise, if no updates where performed, then $(\alpha_t, \beta_t)$ satisfies both $\sum_{i=1}^n \mu_i |\boldeta_i - 1| \leq \eps_2$ and $\sum_{j=1}^m \nu_j |\bxi_j - 1| \leq \eps_2$. In this case, we execute $\EstPenalty(\alpha_t, \beta_t, \eps_1, \eps r^{\rho}, \delta)$ which outputs a number $\bomega \in \R_{\geq 0}$, and we output
\[ \sum_{i=1}^n \mu_i (\alpha_t)_i - \sum_{j=1}^m \nu_j (\beta_t)_j - \bomega. \]
		\end{itemize}
		\end{flushleft}
	\end{framed}
	\caption{Main Algorithm for Estimating $\calR_{\rho}(\mu, \nu)^{\rho}$.} \label{fig:main-alg}
\end{figure}

\subsection{Analysis of the Algorithm}

We now show that the algorithm presented at the top of Subsection~\ref{sec:alg} finds an approximately optimal maximizer of $g$, assuming the lemmas on the guarantees of the subroutines of Subsection~\ref{sec:sub-routine-guarantees}. In particular, this section shows two lemmas. The first lemma shows that if the algorithm does not perform an update, then the value $(\alpha_t, \beta_t)$ that the algorithm holds is an approximate maximizer of $g$, this will then imply, from Lemma~\ref{lem:penalty} that we can output an estimate of $g(\alpha_t, \beta_t)$. The second lemma says that if the algorithm performs an update, then the value of the objective function $g$ increases by $\Omega(\eps_2 \cdot \lambda)$. In particular, since the objective function is a maximization problem which is always at most $r^{\rho}$, this implies that the algorithm performs at most $O(r^{\rho} / (\eps_2 \cdot \lambda))$ updates before it must terminate. In addition, when it terminates, Lemma~\ref{lem:bounded-penalty} implies that the quantity $\bomega$ output by $\EstPenalty$ is at most $O(r^{\rho})$. This means that in the final estimate, for $\bomega$, it suffices to set $\tau$ to $\eps r^{\rho}$. 

\begin{lemma}[Termination Condition]
Suppose $(\alpha_t, \beta_t) \in \R^{n+m}$ satisfies $g(\alpha_t, \beta_t) \geq 0$ and the algorithms $\emph{\EstAlpha}$ and $\emph{\EstBeta}$ produce a sequence of quantities $\boldeta_1,\dots, \boldeta_n$ and $\bxi_1,\dots, \bxi_m$ which satisfy the guarantees of Lemma~\ref{lem:est-alpha} and Lemma~\ref{lem:est-beta}, and
\begin{align*}
\sum_{i=1}^n \mu_i \left| \boldeta_i - 1 \right| \leq \eps_2 \qquad \text{and}\qquad \sum_{j=1}^m \nu_j\left| \bxi_j - 1 \right| \leq \eps_2.
\end{align*}
Then, letting $(\alpha^*, \beta^*)$ be the maximizer of $g(\alpha^*, \beta^*)$, we have 
\[ g(\alpha^*, \beta^*) - g(\alpha_t, \beta_t) \leq O(1) \cdot \left( \frac{nm}{\sigma_{\mu} \sigma_{\nu}}\right)^{(\rho-1)/\rho}\cdot r^{\rho} \cdot \left( \eps_2 + \tau + \frac{\eps_1 s}{\sigma}\right) \]
\end{lemma}

\begin{lemma}[Updates Increase Objective]\label{lem:increase-obj}
Suppose $(\alpha_t, \beta_t) \in \R^{n+m}$ satisfies $g(\alpha_t, \beta_t)\geq 0$ and $(\alpha_{t+1}, \beta_t) \in \R^{n+m}$ is a vector, for which the algorithms $\emph{\EstAlpha}$ produce the sequence of outputs $\boldeta_1,\dots, \boldeta_n$ which satisfy the guarantees of Lemma~\ref{lem:est-alpha} and
\begin{align*}
\sum_{i=1}^n \mu_i \left| \boldeta_i - 1\right| \geq \eps_2.
\end{align*}
Then, $g(\alpha_{t+1}, \beta_t) - g(\alpha_t, \beta_t) \geq \Omega(\lambda \cdot \eps_2)$. 
\end{lemma}

\subsection{Proof of Theorem~\ref{thm:main}}

Consider the algorithm which first runs the preprocessing step of Subsection~\ref{sec:preprocess} and then executes the main iterative sub-routine of Figure~\ref{fig:main-alg} in order to estimate $\calR_{\rho}(\mu, \nu)^{\rho}$. When the algorithm from Figure~\ref{fig:main-alg} outputs, we output
\[ \left( \sum_{i=1}^n \mu_i (\alpha_t)_i - \sum_{j=1}^m \nu_j (\beta_t)_j - \bomega \right)^{1/\rho}. \]
First, we note the running time of the algorithm is as specified. In particular, the preprocessing step takes $O(n+m)$ time. Notice that each iteration of Figure~\ref{fig:main-alg} takes $O(n+m) \cdot \poly^*(2^{s} / \eps_1)$ time, which is $O(n+m) \cdot \poly^*(2^{\rho / (\rho-1)} (mn)^{(\rho-1)/\rho} / \eps)$ by setting of $\eps_1$, $s$, and $\sigma_{\mu}, \sigma_{\nu}$ and $\sigma$. Furthermore, since $g(\alpha_0, \beta_0) = 0$ and $g(\alpha_t, \beta_t) \leq r^{\rho}$, Lemma~\ref{lem:increase-obj} will imply that the number of iterations is at most $O(r^{\rho} / (\lambda \eps_2))$, and by the setting of $\eps_2$ and $\lambda$, this is at most $\poly^*((mn)^{(\rho-1)/\rho} / \eps)$. The total running time then follows.

In order to show correctness, note that the setting of $\eps_2$, $\eps_1$, when the algorithm terminates, we have
\[ \calR_{\rho}(\mu, \nu)^{\rho} - \sum_{i=1}^n \mu_i (\alpha_t)_i - \sum_{j=1}^m \nu_j (\beta_t)_j - C_s \sum_{i=1}^n \sum_{j=1}^m \mu_i \nu_j\left(\frac{((\alpha_t)_i - (\beta_t)_j)^+}{\|x_i - y_j\|_2} \right)^s \leq \eps \cdot r^{\rho},\]
and we are guaranteed by Lemma~\ref{lem:penalty} and Lemma~\ref{lem:bounded-penalty} and the setting of $\tau$ for $\EstPenalty$ that 
\[ \left| \bomega - C_s \sum_{i=1}^n \sum_{j=1}^m \mu_i \nu_j \left( \frac{((\alpha_t)_i - (\beta_t)_j)^+}{\|x_i - y_j\|_2}\right)^{s} \right| \leq O(\eps \cdot r^{\rho}). \]
Therefore, our output (using the fact $\calR_{\rho}(\mu, \nu)^{\rho} \leq r^{\rho}$), will satisfy
\[ \left| \left( \sum_{i=1}^n \mu_i (\alpha_t)_i - \sum_{j=1}^m \nu_j (\beta_t)_j - \bomega\right)^{1/\rho} - \calR_{\rho}(\mu, \nu) \right| \leq  O(\eps \cdot r).\]

\ignore{\section{A Stochastic Gradient Ascent Algorithm}

\newcommand{\err}{\mathrm{err}}
\newcommand{\CheckStepSize}{\texttt{Check-Step-Size}}
\newcommand{\ApproxGradient}{\texttt{ApproxGradient}}

The algorithm will maintain a setting of the dual variables $(\alpha_t, \beta_t) \in \R^{n+m}$, an approximately optimal final dual variables $(\hat{\alpha}_t, \hat{\beta}_t) \in \R^{n+m}$, and a normalization quantity $Z_t \in \R_{\geq 0}$ which it will update in each iteration. We first describe the algorithm assuming access to two algorithms $\CheckStepSize$ and $\ApproxGradient$ which we specify later. We initialize $(\alpha_0, \beta_0) = (0, 0)$ and $(\hat{\alpha}_0, \hat{\beta}_0) = (0, 0)$. We start with a setting of 
\[ T = \frac{1}{4\eps} \qquad\text{and}\qquad \eta = (n+m) \cdot r^{\rho} \cdot \sup_{\substack{i \in [n] \\ j \in [m]}} \left| \frac{\min\{ \mu_i, \nu_j\}}{s C_s \cdot \mu_i \nu_j}\right|^{\rho-1}.\]
For $t =1, \dots, T$, the algorithm performs the following steps:
\begin{enumerate}
\item The algorithm executes $\CheckStepSize(\alpha_{t-1}, \beta_{t-1})$ and produces the output random variable $\blambda_{t} \in \R_{\geq 0}$, if $\eta \leq \eps r^{\rho} / \blambda_t$, we proceed. Otherwise, we update $\eta$ to $\eta / 2$ and $T = 2T$. We then update $t = 1$ and re-start the algorithm.
\item We execute the algorithm $\ApproxGradient(\alpha_{t-1}, \beta_{t-1})$ which returns a random vector $\bDelta_{t} \in \R^{n+m}$. 
\item We perform the following updates:
\begin{align*}
(\alpha_{t}, \beta_{t}) &\longleftarrow (\alpha_{t-1}, \beta_{t-1}) + \eta \cdot \bDelta_{t}. \\
(\hat{\alpha}_{t}, \hat{\beta}_{t}) &\longleftarrow \left(\frac{t}{t + 1}\right) \cdot (\hat{\alpha}_t, \hat{\beta}_t) + \frac{1}{t + 1} \cdot (\alpha_t, \beta_t).
\end{align*}
\end{enumerate}
If we have reached $t = T$, then we output
\[ \left( 1/ \rho \right)^{1/\rho} \left|  \sum_{i=1}^n \mu_i \cdot (\hat{\alpha}_T)_i - \sum_{j=1}^m \nu_j \cdot (\hat{\beta}_T)_j \right|^{1/\rho}.\]

For simplicity in the notation for this section (and in order to give a generic description of the algorithm), we will consider the distributions $\mu$ and $\nu$, supported on $\{x_1,\dots, x_n\}$ and $\{ y_1,\dots, y_m \}$ as fixed. We also may assume (without loss of generality) that $\{x_1,\dots, x_n \}$ and $\{y_1,\dots, y_m \}$ are disjoint by considering an infinitesimally small perturbation of the points. Then, we consider the parameter $\rho > 1$ and the H\"{o}lder conjugate $s \geq 1$, i.e., $1/\rho + 1/s = 1$, and write, for any $(\alpha, \beta) \in \R^{n+m}$, 
\[ F_{ij}(\alpha, \beta) = \left( \dfrac{(\alpha_i + \beta_j)^+}{\|x_i - y_j \|_2} \right)^{s} \qquad \text{and}\qquad f_{ij}(\alpha, \beta) = s \cdot \dfrac{((\alpha_i + \beta_j)^+)^{s-1}}{\|x_i - x_j\|_2^s}, \]
and the function
\[ g(\alpha, \beta) = \sum_{i=1}^n \mu_i \alpha_i + \sum_{j=1}^m \nu_j \beta_j - C_s \sum_{i=1}^n \sum_{j=1}^m \mu_i \nu_j \cdot F_{ij}(\alpha, \beta), \]
where $C_s = (1/s) (1-1/s)^{s-1}$.
\begin{lemma}[Convergence Guarantee]
Suppose that the algorithms $\emph{\ApproxGradient}(\alpha, \beta)$ and $\emph{\CheckStepSize}(\alpha, \beta)$ satisfy the following guarantees:
\begin{itemize}
\item The algorithm $\emph{\ApproxGradient}(\alpha, \beta)$ is a randomized algorithm which outputs a random vector $\bDelta \in \R^{n + m}$ whose expectation is the gradient $\nabla g(\alpha, \beta)$.
\item The randomized algorithm $\emph{\CheckStepSize}(\alpha, \beta)$ outputs a value $\blambda \in \R_{\geq 0}$ which satisfies
\[ \Ex_{\bDelta}\left[ \left\| \bDelta \right\|_2^2\right] \leq \blambda \]
with probability $1-1/\poly(nm)$ when $\bDelta$ is generated according to $\emph{\ApproxGradient}(\alpha, \beta)$.
\end{itemize}
Then, if $(\alpha^*, \beta^*)$ is the maximizer of $g$, whenever the algorithm reaches $t = T$ and produces an output, we satisfy 
\begin{align*}
\Ex\left[ g(\alpha^*, \beta^*) - g(\hat{\balpha}_t, \hat{\bbeta}_t) \right] \leq \left(\dfrac{(n + m) \cdot r^{\rho}}{8\cdot t \cdot \eta} \cdot \sup_{\substack{i \in [n] \\ j \in [m]}} \left| \frac{\min\{ \mu_i, \nu_j\}}{s C_s \cdot \mu_i\nu_j} \right|^{2(\rho-1)} + \eps\right)  r^{\rho}.
\end{align*}
\end{lemma}

\begin{proof}
We will follow the usual analysis of stochastic gradient descent. For ease in notation, we will let $z \in \R^{n+m}$ where the first $n$ coordinates index into $\alpha$ and the second $m$ into $\beta$; this way we may write $z_t = (\alpha_t, \beta_t)$ and $z^* = (\alpha^*, \beta^*)$. In particular, for any fixed execution of iterations $1,\dots, t-1$, we have
\begin{align*}
g(z^*) - g(z_{t-1}) &\leq \langle \nabla g(z_{t-1}) , z^* - z_{t-1} \rangle = \Ex_{\bDelta_t}\left[ \langle \bDelta_t, z^* - z_{t-1} \rangle \right] = \frac{1}{\eta} \cdot \Ex_{\bDelta_t}\left[ \langle \bz_t - z_{t-1}, z^* - z_{t-1}\rangle \right] \\
		&= \frac{1}{2 \cdot \eta} \cdot \left( \| z^* - z_{t-1} \|_2^2 - \Ex_{\bDelta_t}\left[\| z^* - \bz_{t}\|_2^2\right]\right)  + \frac{\eta}{2} \Ex_{\bDelta_t}\left[ \| \bDelta_t \|_2^2 \right],
\end{align*}
which implies that for any sequence $z_1,\dots, z_{t-1}$, we can write
\begin{align*} 
&g(z^*) - g\left( \frac{1}{t} \sum_{i=1}^t z_{i-1}\right) \leq \frac{1}{t} \sum_{i=1}^t (g(z^*) - g(z_{i-1})) \\
			&\qquad\qquad\leq \frac{1}{2t \cdot \eta} \sum_{i=1}^t \|z^*- z_{i-1}\|_2^2- \frac{1}{2t \cdot \eta} \sum_{i=1}^t \Ex_{\bDelta_i'} \left[\|z^* - \bz_i'\|_2^2 \right] + \frac{1}{2t} \sum_{i=1}^t \eta \cdot \Ex_{\bDelta_i'}\left[ \| \bDelta_i'\|_2^2\right].
\end{align*}
which means that if we always have $\eta \leq 2\eps \cdot r^{\rho} / \blambda_i \leq 2\eps \cdot r^{\rho} / \Ex[\|\bDelta_i'\|_2^2]$, and the guarantees of $\CheckStepSize(\alpha_t, \beta_t)$ hold, the above quantity is bounded by
\begin{align*}
g(z^*) - g\left( \frac{1}{t} \sum_{i=1}^t z_{i-1}\right)  &\leq \frac{1}{2t \cdot \eta} \sum_{i=1}^t \|z^*- z_{i-1}\|_2^2- \frac{1}{2t \cdot \eta} \sum_{i=1}^t \Ex_{\bDelta_i'} \left[\|z^* - \bz_i'\|_2^2 \right] + \eps \cdot r^{\rho},
\end{align*}
which implies by linearity of expectation and the telescoping sum, that 
\begin{align*}
\Ex_{\bDelta_1,\dots, \bDelta_t}\left[ g(z^*) - g\left( \frac{1}{t} \sum_{i=1}^t z_{i-1}\right)\right] \leq \frac{\| z^*\|_2^2}{2t \cdot \eta} + \eps \cdot r^{\rho}.
\end{align*}
The final bound follows from the upper bound on $\| (\alpha^*, \beta^*)\|_{2}^2 \leq (n+m) \|(\alpha^*, \beta^*)\|_{\infty}^2$.
\end{proof}

\newcommand{\out}{\mathtt{out}}

\begin{lemma}[Output Guarantee]
Suppose $(\hat{\alpha}, \hat{\beta}) \in \R^{n+m}$ is the final point in the algorithm, which satisfies the following three guarantees for a small parameter $\eps \in (0, 1/4)$,
\begin{itemize}
\item First, we have $g(\alpha^*, \beta^*) - g(\hat{\alpha}, \hat{\beta}) \leq \eps \cdot r^{\rho}$.
\item Second, we have $\| (\hat{\alpha}, \hat{\beta})\|_2 \cdot \| \nabla g(\hat{\alpha}, \hat{\beta}) \|_2 \leq s \cdot \eps r^{\rho}$. 
\end{itemize}
Then, if we let $\out \in \R_{\geq 0}$ be defined as
\[ \out = \left(1/\rho \right)^{1/\rho} \left| \sum_{i=1}^n \mu_i \hat{\alpha}_i - \sum_{j=1}^m \nu_j \hat{\beta}_j\right|^{1/\rho}, \]
we have
\begin{align*} 
&\left| \out - \calR_{\rho}(\mu, \nu)\right| \leq O(\eps r / \rho).
\end{align*}
\end{lemma}

\begin{proof}
First, note that $\calR_{\rho}(\mu, \nu)^{\rho} = g(\alpha^*, \beta^*)$. Then, we note that we may re-write
\begin{align*}
g(\hat{\alpha}, \hat{\beta}) &= \sum_{i=1}^n \mu_i \hat{\alpha}_i - \sum_{j=1}^m \nu_j \hat{\beta}_j \\
	&\qquad \qquad - \frac{1}{s} \sum_{i=1}^n \mu_i \hat{\alpha}_i \cdot s C_s \sum_{j=1}^m \nu_j \cdot \frac{((\hat{\alpha}_i - \hat{\beta}_j)^+)^{s-1}}{\|x_i - y_j\|_2^s}\\
	&\qquad\qquad + \frac{1}{s} \sum_{j=1}^m \nu_j \hat{\beta}_j \cdot s C_s \sum_{i=1}^n \mu_i \cdot \frac{((\hat{\alpha}_i - \hat{\beta}_j)^+)^{s-1}}{\|x_i - y_j\|_2^s}\\
	&= \left(1 - \frac{1}{s} \right) \left( \sum_{i=1}^n \mu_i \hat{\alpha}_i - \sum_{j=1}^m \nu_j \hat{\beta}_j \right) + \frac{1}{s} \sum_{i=1}^n \hat{\alpha}_i \cdot \nabla g(\hat{\alpha}, \hat{\beta})_i - \frac{1}{s} \sum_{j=1}^m \hat{\beta}_j \cdot \nabla g(\hat{\alpha},\hat{\beta})_j \\
	&= \out^{\rho} + \frac{1}{s} \sum_{i=1}^n \hat{\alpha}_i \nabla g(\hat{\alpha}, \hat{\beta}) - \frac{1}{s} \sum_{j=1}^m \hat{\beta}_j \nabla g(\hat{\alpha}, \hat{\beta})_j
\end{align*} 
By the Cauchy-Schwarz inequality, and the fact $g(\hat{\alpha}, \hat{\beta})$ and $g(\alpha^*, \beta^*)$ differ by at most $\eps r^{\rho}$, we have
\begin{align*}
\left|\calR_{\rho}(\mu, \nu)^{\rho} - \out^{\rho} \right| \leq 2\eps \cdot r^{\rho} 
\end{align*}
Thus, we conclude by realizing that
\begin{align*}
\left| \frac{\calR_{\rho}(\mu, \nu) - \out}{r} \right| &= \left| \left(\left(\frac{\calR_{\rho}(\mu, \nu)}{r}\right)^{\rho}\right)^{1/\rho} - \left(\left(\frac{\out}{r} \right)^{\rho}\right)^{1/\rho} \right| \\
	&\leq \left| \left(\frac{\calR_{\rho}(\mu, \nu)}{r}\right)^{\rho} - \left(\frac{\out}{r}\right)^{\rho} \right| \cdot \frac{1}{\rho} \cdot \min\left\{ \left(\frac{\calR_{\rho}(\mu, \nu)}{r}\right)^{\rho}, \left(\frac{\out}{r}\right)^{\rho}\right\}^{1/\rho - 1} \\
	&\leq \frac{2 \eps}{\rho} \left( 1 - \frac{1}{r^{\rho}}\left| \calR_{\rho}(\mu,\nu)^{\rho} - \out^{\rho} \right| \right)^{1/\rho - 1} \leq \frac{2\eps}{\rho} \cdot \left( 1 - 2\eps\right)^{1/\rho - 1}.
\end{align*}
\end{proof}}

\section{Open Problems}

We hope that our approach, of slightly changing the problem, will prove useful for other Euclidean problems for which we do not have fast algorithms with $(1+\eps)$-approximations. We mention two immediate open problems:
\begin{itemize}
\item Multiplicative $(1+\eps)$-approximations for $\mathcal{R}_{\rho}(\mu, \nu)$. Our algorithms achieved additive $\eps r$-approximations for datasets bounded within distance $r$, but a more accurate multiplicative $(1+\eps)$-approximation would be desired when the dataset may not necessarily be bounded. Does there exists an algorithm which is just as fast as Theorem~\ref{thm:main} and outputs a number $\hat{\boldeta}$ which is between $\mathcal{R}_{\rho}(\mu, \nu)$ and $(1+\eps) \mathcal{R}_{\rho}(\mu, \nu)$ with high probability?
\item Accurate Approximations for $\EMD$. It is still possible that for any $\eps > 0$, there exists an algorithm which can estimate the cost of $\EMD(\mu, \nu)$ up to a multiplicative $(1+\eps)$-factor in time $n \cdot \poly(d \log n/\eps)$. Does there exist such an algorithm, or is there compelling complexity-theoretic reasons why this may not be possible? We note that~\cite{R19} shows that, in the case $\mu$ and $\nu$ are uniform on a support of size $n$, such an algorithm should not be able to output a $(1+\eps)$-approximate matching between points of $\mu$ and $\nu$ (assuming the Hitting Set conjecture). However, no such evidence against near-linear time algorithms for the \emph{cost} of $\EMD$ exists.
\end{itemize}

\acks{Part of this work was done while Erik Waingarten was a postdoc at Stanford University, supported by an NSF postdoctoral fellowship and by Moses Charikar's Simons Investigator Award.}

\bibliography{waingarten}

\appendix

\section{Basic Properties of $\calR_{\rho}(\mu, \nu)$}\label{sec:basic-prop}

The main property of $\calR_{\rho}(\mu, \nu)$ aligning with the Earth Mover's Distance is that $\calR_{\rho}(\mu, \nu)$ defines a metric space over the space of distributions. As is the case with the Earth Mover's distance, the metric takes into account the underlying geometry of the space. The advantage to these formulations, as opposed to, say the total variation distance, is that the total variation distance between a distribution $\mu$ and that same distribution $\mu'$ after a tiny perturbation of the points is $1$; whereas the Earth Mover's Distance and $\calR_{\rho}(\mu, \mu')$ will remain small.

\subsection{$\calR_{\rho}(\mu, \nu)$ Defines a Metric Space}

The one non-trivial aspect of showing $\calR_{\rho}(\cdot, \cdot)$ is a metric over distributions in $\R^d$ is establishing the triangle inequality. Indeed, the definition $\calR_{\rho}(\mu, \nu) = \calR_{\rho}(\nu, \mu)$ by symmetry in the definition, and it is clear that $\calR_{\rho}(\mu, \nu) \geq 0$ and equal to zero whenever $\mu$ and $\nu$ are equal. 

\begin{lemma}[Triangle Inequality]\label{lem:triangle}
Suppose $\mu, \nu$ and $\xi$ are discrete distributions supported on vectors in $\R^d$. Then, for any $\rho \geq 1$, 
\begin{align*}
\calR_{\rho}(\mu, \xi) \leq \calR_{\rho}(\mu, \nu) + \calR_{\rho}(\nu,\xi).
\end{align*}
\end{lemma}

\begin{proof}
Let $\{ x_1,\dots, x_n \}$ denote the set of points in the supports on $\mu, \nu$ and $\xi$. We will then use $\mu, \nu, \xi$ as vectors in $\R^n_{\geq 0}$ whose coordinates sum to $1$ and specify with which probability to sample the point $x_i$. Suppose that $\gamma^{(1)}, \gamma^{(2)} \in \R^{n\times n}_{\geq 0}$ are the couplings which minimize $\calR_{\rho}(\mu, \nu)$ and $\calR_{\rho}(\nu, \xi)$, respectively. Then, consider setting
\[ \gamma_{ij} = \sum_{\ell=1}^n \frac{\gamma^{(1)}_{i \ell} \cdot \gamma^{(2)}_{\ell j}}{\nu_{\ell}}. \]
First, one can easily verify that $\gamma$ satisfies the constraints of $\calR_{\rho}(\mu, \xi)$, so that it is a coupling of $\mu$ and $\xi$. The triangle inequality then comes from using the triangle inequality in the ground metric ($\ell_2$), and then the triangle inequality for the $\ell_{\rho}$-norm encoding the cost. Specifically, 
\begin{align}
\calR_{\rho}(\mu,\xi) &\leq \left(\sum_{i=1}^n \sum_{j=1}^n \mu_i \xi_j \left(\frac{\gamma_{ij}}{\mu_i \xi_j} \|x_i - x_j\|_2\right)^{\rho} \right)^{1/\rho} 
\nonumber \\		
		&\leq \left(\sum_{i=1}^n \sum_{j=1}^n \mu_i \xi_j \left(\frac{1}{\mu_i \xi_j\nu_{\ell}} \sum_{\ell=1}^n \gamma^{(1)}_{i \ell} \gamma^{(2)}_{\ell j} \|x_i - x_{\ell}\|_2 \right)^{\rho} \right)^{1/\rho} \label{eq:term-1}\\
		&\qquad\qquad\qquad +  \left(\sum_{i=1}^n \sum_{j=1}^n \mu_i \xi_j \left(\frac{1}{\mu_i \xi_j\nu_{\ell}} \sum_{\ell=1}^n \gamma^{(1)}_{i \ell} \gamma^{(2)}_{\ell j} \|x_{\ell} - x_{j}\|_2 \right)^{\rho} \right)^{1/\rho}. \label{eq:term-2}
\end{align}
We work on each term individually. In particular, for each $j \in [n]$, we use Jensen's inequality, applied to the distribution $\tilde{\gamma}^{(2)}_j$ which is supported on $[n]$ and samples $\bell =\ell$ with probability $\gamma_{\ell j}^{(2)} / \xi_j$ can write
\begin{align*}
\sum_{j=1}^n \xi_j \left( \frac{1}{\mu_i \xi_j\nu_{\ell}} \sum_{\ell=1}^n \gamma_{i \ell}^{(1)} \gamma^{(2)}_{\ell j} \cdot \|x_i - x_{\ell}\|_2\right)^{\rho} &=  \sum_{j=1}^n \xi_j \left( \Ex_{\bell \sim \tilde{\gamma}^{(2)}_j}\left[ \frac{\gamma^{(1)}_{i \bell}}{\mu_i \nu_{\bell}} \cdot \| x_{i} - x_{\bell}\|_2\right]\right)^{\rho} \\
			&\leq \sum_{j=1}^n \xi_j \sum_{\ell=1}^{n} \frac{\gamma^{(2)}_{\ell j}}{\xi_j} \cdot \left( \frac{\gamma_{i\ell}^{(1)}}{\mu_i\nu_{\ell}}\cdot \|x_i - x_{\ell}\|_2\right)^{\rho} \\
			&= \sum_{\ell=1}^n \nu_{\ell} \left(\frac{\gamma_{i\ell}^{(1)}}{\mu_i \nu_{\ell}} \cdot \|x_i - x_{\ell}\|_2 \right)^{\rho}.
\end{align*}
Applied to each $i \in[n]$, we have that (\ref{eq:term-1}) is at most $\calR_{\rho}(\mu, \nu)$. Similarly, we have (\ref{eq:term-2}) is at most $\calR_{\rho}(\nu, \xi)$.
\end{proof}

\subsection{Dual of $\calR_{\rho}(\mu, \nu)^{\rho}$}

As discussed, the advantage of $\calR_{\rho}(\mu, \nu)$ is its computational properties, since we will give an algorithm to additively approximate $\calR_{\rho}(\mu, \nu)$ efficiently (for small $\rho$). The algorithm will proceed by optimizing over the dual formulation of $\calR_{\rho}(\mu, \nu)^{\rho}$. In this subsection, we specify what the dual of $\calR_{\rho}(\mu, \nu)^{\rho}$ looks like, as well as a few important properties of the dual which we will use throughout the execution of  the algorithm. In the next lemma, it will be important that the support of the distributions $\mu$ and $\nu$ are disjoint (so that we never divide by zero); in our algorithm, we will apply a preprocessing step so as to assume without loss of generality, that the supports of $\mu$ and $\nu$ are disjoint.

\begin{lemma}
Let $\mu \in \R^n_{> 0}$ encode a distribution supported on $\{x_1,\dots, x_n \}\subset \R^d$ and $\nu \in \R^{m}_{>0}$ encode a distribution supported on $\{y_1,\dots, y_m \} \subset \R^d$. Whenever $\{x_1,\dots, x_n\}$ and $\{ y_1, \dots, y_m\}$ are disjoint, the following is true for any $\rho > 1$. Let $s \geq 1$ denote its H\"{o}lder conjugate, $1/\rho + 1/s = 1$, and $C_s = (1/s) (1 - 1/s)^{s-1}$. Then,
\begin{align}
\calR_{\rho}(\mu, \nu)^{\rho} &= \max_{\substack{\alpha \in \R^n \\ \beta \in \R^m}} \sum_{i=1}^n \mu_i \alpha_i - \sum_{j=1}^m \nu_j \beta_j - C_s \sum_{i=1}^n \sum_{j=1}^m \mu_i \nu_j \left( \frac{(\alpha_i - \beta_j)^+}{\|x_i - y_j \|_2}\right)^s. \label{eq:dual}
\end{align}
\ignore{Furthermore, there exists an optimal solution $(\alpha^*, \beta^*)$ with 
\[ \| (\alpha^*, \beta^*) \|_{\infty} \leq \frac{1}{2} \cdot \sup\left| \frac{\min\{ \mu_i, \nu_j \}}{s\cdot C_s \cdot \mu_i \nu_j} \right|^{\rho-1} r^{\rho}. \]
\ignore{and at any optimal solution $(\alpha^*, \beta^*)$ we have 
\[ \calR_{\rho}(\mu, \nu)^{\rho} = (1/\rho) \left(\sum_{i=1}^n \mu_i \alpha_i^* - \sum_{j=1}^m \nu_j \beta_j^* \right). \]}}
\end{lemma}

\begin{proof}
We compute the dual of $\calR_{\rho}(\mu, \nu)^{\rho}$ by introducing Lagrangian multipliers.
\begin{align*}
\calR_{\rho}(\mu, \nu)^{\rho} = \max_{\substack{\alpha \in \R^n \\ \beta \in \R^m}} \min_{\gamma \in \R^{n\times m}_{\geq 0}} &\sum_{i=1}^n \sum_{j=1}^m \mu_i\nu_j \left(\frac{\gamma_{ij}}{\mu_i\nu_j} \right)^{\rho} \|x_i - y_j\|_2^{\rho} \\
	&\qquad + \sum_{i=1}^n \alpha_i \left(\mu_i - \sum_{j=1}^m \gamma_{ij}\right)  - \sum_{j=1}^m \beta_j \left( \nu_j - \sum_{i=1}^n \gamma_{ij}\right),
\end{align*}
and we re-write this, by minimax duality
\begin{align*}
\min_{\gamma} \max_{\alpha,\beta} \sum_{i=1}^n \mu_i \alpha_i + \sum_{j=1}^m \nu_j \beta_j + \sum_{i=1}^n \sum_{j=1}^m \left( (\mu_i \nu_j)^{1-\rho} \gamma_{ij}^{\rho} \|x_i - y_j\|_2^{\rho}  - (\alpha_i - \beta_j) \gamma_{ij} \right).
\end{align*}
Since the right-most term is the only one to depend on $\gamma$, once we fix $\alpha, \beta$, we can compute the minimizing value $\gamma_{ij} \geq0$ for each term $i,j$. The expression for a (generic) term $i,j$ considers a positive value $a > 0$ (the fact the term is positive is by disjointness of $\{x_1,\dots,x_n\}$ and $\{y_1,\dots, y_m\}$) and some term $b \in \R$, and we want to determine
\[ \min_{\gamma_{ij} \geq 0} \gamma_{ij}^{\rho} a - \gamma_{ij} b. \]
Note that if $b$ is negative, then the minimizing value occurs at $0$ with $\gamma_{ij} = 0$. On the other hand, if $b$ is positive, then we want to set $\gamma_{ij}$ such that the derivative with respect to $\gamma_{ij}$ is zero, since this corresponds to the (positive) value of $\gamma_{ij}$ where we get a negative value of the optimum above:
\[ \rho \gamma_{ij}^{\rho - 1} a - b = 0  \Longrightarrow  \gamma_{ij} = \left( \frac{b}{\rho a}\right)^{1/(\rho-1)}  \Longrightarrow  \gamma_{ij}^{\rho} a - \gamma_{ij} b = \dfrac{b^{\rho/(\rho-1)}}{a^{1/(\rho-1)}}\cdot\left( \frac{1}{\rho^{\rho/(\rho-1)}} - \frac{1}{\rho^{1/(\rho-1)}} \right)\]
In order to simplify the notation, let $s \geq 1$ denote the H\"{o}lder conjugate of $\rho$. So that
\[ \frac{1}{\rho} + \frac{1}{s} = 1 \qquad \frac{\rho}{\rho-1} = s \qquad \frac{1}{\rho-1} = s - 1, \]
and the above expression becomes:
\begin{align*}
-\left( \frac{b}{a}\right)^{s} \cdot a \cdot \frac{1}{s} \left(1 - \frac{1}{s} \right)^{s-1}.
\end{align*}
Plugging this in to each term, and letting $C_s = (1/s) (1 - 1/s)^{s-1}$, each term becomes
\begin{align*} 
C_s \left( \left(\dfrac{1}{\mu_i \nu_j}\right)^{(1-\rho)} \dfrac{(\alpha_i -\beta_j)^+}{\|x_i - y_j\|_2^{\rho}}\right)^{s} (\mu_i \nu_j)^{1-\rho} \|x_i - y_j\|_2^{\rho} &= C_s \cdot \mu_i \nu_j \left(\frac{(\alpha_i - \beta_j)^+}{\|x_i - y_j \|_2}\right)^{s}.
\end{align*}
which means the dual becomes
\[ \max_{\alpha,\beta \in \R^n} \sum_{i=1}^n \mu_i \alpha_i - \sum_{j=1}^n \nu_j \beta_j - C_s \sum_{i=1}^n \sum_{j=1}^n \mu_i \nu_j \left(\dfrac{(\alpha_i - \beta_j)^+}{\|x_i - x_j \|_2} \right)^s, \]
and for any $\alpha, \beta \in \R^n$, the minimizing value of $\gamma$ is
\[ \gamma_{ij} = s C_s \cdot \mu_i \nu_j \cdot \dfrac{((\alpha_i - \beta_j)^+)^{s-1}}{\|x_i - x_j \|_2^s}.\]
\ignore{To upper bound the $\ell_{\infty}$-norm of any optimal dual solution $(\alpha^*, \beta^*)$, first note that in the maximizing $(\alpha^*, \beta^*)$, we always has the smallest value of $\alpha^*$ being greater than or equal to the smallest term of $\beta^*$ (since we can otherwise increase it to increase the objective), and the largest term of $\beta^*$ is always smaller than or equal to the largest term in $\alpha^*$ (since otherwise, we can decrease it to increase the objective). In particular, all entries lie between the largest value of $\alpha^*$ and the smallest value of $\beta^*$. In addition, if $i \in [n]$ is the coordinate corresponding to the largest entry of $\alpha^*$, and $j \in [m]$ is the coordinate corresponding to the smallest entry of $\beta^*$, since $\gamma_{ij} \leq \min \{ \mu_i, \nu_j\}$, we must have
\[ (\alpha_i^* - \beta_j^*)^+ \leq \left| \frac{\min\{ \mu_i, \nu_j\}}{s \cdot C_s \cdot \mu_i \nu_j} \right|^{1/(s-1)} r^{s/(s-1)}. \]
In particular, we may translate all entries of $\alpha^*$ and $\beta^*$ in order to have 
\[\| (\alpha^*, \beta^*) \|_{\infty} \leq \frac{1}{2} \cdot \sup_{\substack{i\in[n] \\ j \in [m]}}\left|\frac{\min\{ \mu_i, \nu_j\}}{s \cdot C_s \cdot \mu_i \nu_j} \right|^{\rho-1} \cdot r^{\rho}. \]
\ignore{Finally, we have
\begin{align*}
&\sum_{i=1}^n \mu_i \alpha_i^* - \sum_{j=1}^m \nu_j \beta_j^* - C_s \sum_{i=1}^n \sum_{j=1}^m \mu_i \nu_j (\alpha_i^* - \beta_j^*) \cdot \dfrac{((\alpha_i^* - \beta_j^*)^+)^{s-1}}{\|x_i - y_j\|_2^s} =  \\
&\qquad = \sum_{i=1}^n \mu_i \alpha_i^* - \sum_{j=1}^m \nu_j \beta_j^* - \frac{1}{s} \sum_{i=1}^n \mu_i \alpha_i^* \left( s C_s \sum_{j=1}^m \nu_j \cdot \dfrac{((\alpha_i^* - \beta_j^*)^+)^{s-1}}{\|x_i - y_j\|_2^s}\right) \\
&\qquad\qquad\qquad + \frac{1}{s} \sum_{j=1}^m \nu_j \beta_j^* \left( s C_s \sum_{i=1}^n \mu_i \cdot \frac{((\alpha_i^* - \beta_j^*)^+)^{s-1}}{\|x_i - y_j\|_2^s}\right) \\
&\qquad = \left(1 - \frac{1}{s} \right) \left( \sum_{i=1}^n \mu_i \alpha_i^* - \sum_{j=1}^m \nu_j \beta_j^* \right).
\end{align*}}}
\end{proof}

We intuitively think of the dual of $\calR_{\rho}(\mu, \nu)^{\rho}$ in (\ref{eq:dual}) as consisting of two parts. The linear term, the first two summands in (\ref{eq:dual}), and a non-linear ``penalty'' term given by $C_s \sum_{i=1}^n \sum_{j=1}^m \mu_i \nu_j ((\alpha_i - \beta_j)^+/\|x_i - y_j\|_2)^s$. It is useful to compare the dual of $\calR_{\rho}(\mu, \nu)^{\rho}$ to the dual of the Earth Mover's distance, which would essentially only consider the two linear terms and always enforce that $\alpha_i - \beta_j \leq \|x_i - y_j\|_2$. In (\ref{eq:dual}), the $nm$ constraints of the Earth Mover's distance instead become a single ``penalty'' which encode how much larger $\alpha_i - \beta_j$ is, in comparison to $\|x_i - y_j\|_2$. The following two lemmas will become important in the analysis of the algorithm, since they will say that the non-linear ``penalty'' term does not become too large throughout the optimization. Intuitively, this will mean that as we perform a gradient ascent algorithm, the function which we optimize remains smooth.

\begin{lemma}\label{lem:bounded-penalty}
Let $s \geq 2$ and $g \colon \R^n \times \R^m \to \R_{\geq 0}$ denote the objective function which we seek to optimize, i.e., that which is the dual of $\calR_{\rho}(\mu, \nu)^{\rho}$,
\[ g(\alpha, \beta) = \sum_{i=1}^n \mu_i \alpha_i - \sum_{j=1}^m \nu_j \beta_j - C_s \sum_{i=1}^n \sum_{j=1}^m \mu_i \nu_j \left( \dfrac{(\alpha_i - \beta_j)^+}{\|x_i - y_j\|_2}\right)^{s}.\]
Then, if $(\alpha, \beta) \in \R^{n+m}$ is any point where $g(\alpha,\beta) \geq 0$, then
\[ C_s \sum_{i=1}^n \sum_{j=1}^m \mu_i \nu_j \left( \frac{(\alpha_i - \beta_j)^+}{\|x_i - y_j\|_2}\right)^s \leq 4 \cdot \calR_{\rho}(\mu, \nu)^{\rho}. \]
In addition, there exists an input $(\alpha', \beta') \in \R^{n+m}$ with $g(\alpha', \beta') \geq g(\alpha, \beta)$ which satisfies
\[ \|(\alpha', \beta') \|_{\infty} \leq \frac{4^{1/s}}{2} \cdot \sup_{i,j} \left( \frac{1}{\mu_i \nu_j}\right)^{(\rho-1)/\rho} \cdot r^{\rho}. \]
\end{lemma}

\begin{proof}
Let $\psi$ denote the quantity
\[ \psi = C_s  \sum_{i=1}^n \sum_{j=1}^m \mu_i \nu_j \left( \frac{(\alpha_i - \beta_j)^+}{\|x_i - y_j\|_2}\right)^s \]
which we wish to upper bound. Then, by the fact $g(\alpha, \beta) \geq 0$, we must have that
\[ \zeta = \sum_{i=1}^n \mu_i \alpha_i - \sum_{j=1}^m \nu_j \beta_j \geq \psi. \]
If we consider the point $(\alpha', \beta') \in \R^{n+m}$ given by $(\alpha, \beta) / 2$, we have $g(\alpha', \beta') = \zeta / 2 - \psi / 2^s$, and since the dual of $\calR_{\rho}(\mu, \nu)^{\rho}$ is a maximization problem, we must have
\[ \psi \left(\frac{1}{2} - \frac{1}{2^s}\right) \leq \frac{\zeta}{2} - \frac{\psi}{2^s} = g(\alpha', \beta') \leq \calR_{\rho}(\mu, \nu)^{\rho}. \]
When $s \geq 2$, we may then obtain the desired upper bound. For the second part of the claim, we note that given $(\alpha, \beta)$, we can let $(\alpha', \beta')$ be the vector where $\alpha_i' = \max\{ \alpha_i, \min_j \beta_j\}$ and $\beta_j' = \min\{ \beta_j, \max_i \alpha_i'\}$. This means that $g(\alpha', \beta')$ cannot be smaller than $g(\alpha, \beta)$, since we have not increased the penalty $\psi$, while potentially increasing $\zeta$. This means that every entry of $\alpha'$ and $\beta'$ lies between the smallest entry of $\beta'$ and the largest entry of $\alpha'$. However, we note that $\max_i \alpha_i' - \min_j \beta_j'$ can be at most 
\[ 4^{1/s} \cdot r^{\rho} \cdot \sup_{i \in [n], j \in [m]} \left(\frac{1}{\mu_i \nu_j}\right)^{1/s} , \]
before the penalty term $\psi$ exceeds $4 \cdot \calR_{\rho}(\mu, \nu)^{\rho} \leq 4 r^{\rho}$. Thus, we can translate every entry of $\alpha'$ by subtracting $c$ and every entry of $\beta'$ by adding $c$ without changing the objective function, and ensuring the bound on the $\ell_{\infty}$-norm of $\alpha'$ and $\beta'$.
\end{proof}

\begin{lemma}\label{lem:bounded-penalty-grad}
Let $\rho \in (1, 2)$ and $s > 2$ be the H\"{o}lder conjugate, and $g$ as in Lemma~\ref{lem:bounded-penalty}. Then, if $(\alpha, \beta) \in \R^{n+m}$ where $g(\alpha,\beta) \geq 0$, then if $\rmin = \min_{i, j} \|x_i - y_j\|_2$ and $\max_{i, j} \|x_i - y_j\|_2 \leq r$, we have
\begin{align*}
C_s \sum_{i=1}^n \sum_{j=1}^m \mu_i \nu_j \cdot \frac{((\alpha_i - \beta_j)^+)^{s-1}}{\|x_i - y_j\|_2^s} \leq \frac{4\cdot r}{\rmin} \qquad \text{and}\qquad C_s \sum_{i=1}^n \sum_{j=1}^m \mu_i \nu_j \cdot \frac{((\alpha_i - \beta_j)^+)^{s-2}}{\|x_i - y_j\|_2^{s}} \leq \frac{4\cdot r}{\rmin} 
\end{align*}
\end{lemma}

\begin{proof}
Both upper bounds follow from Jensen's inequality and Lemma~\ref{lem:bounded-penalty}. Namely, we use the fact that the function $(\cdot)^{s / (s-1)}$ is convex, and that $\rho = s/(s-1)$,
\begin{align*}
C_s \sum_{i=1}^n \sum_{j=1}^m \mu_i \nu_j \cdot \frac{((\alpha_i - \beta_j)^+)^{s-1}}{\|x_i - y_j\|_2^s} &\leq \frac{C_s}{\rmin} \left( \sum_{i=1}^n \sum_{j=1}^m \mu_i \nu_j \cdot \left(\frac{(\alpha_i - \beta_j)^+}{\|x_i - y_j\|_2} \right)^{s}\right)^{1/\rho} \\
	&\leq \frac{C_s}{\rmin} \left(\frac{4\cdot r^{\rho}}{C_s} \right)^{1/\rho} \leq \frac{4 \cdot r}{\rmin}.
\end{align*}
Similarly, $(\cdot)^{s / (s-2)}$ is convex, and $\rho (s-2) / s = 2 - \rho$, 
\begin{align*}
C_s \sum_{i=1}^n \sum_{j=1}^m \mu_i \nu_j \cdot \frac{((\alpha_i - \beta_j)^+)^{s-2}}{\|x_i - y_j\|_2^s} &\leq \frac{C_s}{\rmin^2} \left( \sum_{i=1}^n \sum_{j=1}^m \mu_i \nu_j \cdot \left(\frac{(\alpha_i - \beta_j)^+}{\|x_i - y_j\|_2} \right)^{s}\right)^{(s-2)/ s} \\
	&\leq \frac{C_s}{\rmin^2} \left(\frac{4\cdot r^{\rho}}{C_s} \right)^{(s-2)/s} \leq \frac{4 \cdot r^{2-\rho}}{\rmin^2}.
\end{align*}
\end{proof}

\ignore{\subsection{Relations Between $\calR_{\rho}(\mu, \nu)$ and $\calW_1(\mu, \nu)$}

\begin{lemma}\label{lem:relate-w-1}
Let $\mu, \nu \in \R^n_{\geq 0}$ encode any two discrete distributions supported on $\{x_1,\dots,x_n\} \subset \R^d$. Then,
\begin{align*}
\calW_1(\mu,\nu) \leq \calR_{\rho}(\mu, \nu) \leq \sup_{i, j \in [n]} \left| \frac{1}{\mu_i \nu_j}\right|^{(\rho-1)/\rho} \calW_{1}(\mu, \nu). 
\end{align*}
\end{lemma}

\begin{proof}
Suppose $\gamma \in \R^{n\times n}_{\geq 0}$ is the minimizing coupling for $\calR_{\rho}(\mu, \nu)$. By Jensen's inequality, we lower bound $\calR_{\rho}(\mu, \nu)^{\rho}$ by
\begin{align*}
\Ex_{\substack{\bi \sim \mu \\ \bj \sim \nu}} \left[ \left(\frac{\gamma_{\bi\bj}}{\mu_{\bi} \nu_{\bj}} \right)^{\rho} \cdot \| x_{\bi} - x_{\bj}\|_2^{\rho}\right] &\geq \left( \Ex_{\substack{\bi \sim \mu \\ \bj \sim \nu}}\left[ \frac{\gamma_{\bi\bj}}{\mu_{\bi} \nu_{\bj}} \cdot \| x_{\bi} - x_{\bj}\|_2 \right] \right)^{\rho} = \left( \sum_{i=1}^n \sum_{j=1}^n \gamma_{ij} \|x_i - x_j\|_2 \right)^{\rho} \\
&= \calW_{1}(\mu, \nu)^{\rho}.
\end{align*}
On the other hand, if $\gamma$ is the minimizing coupling for $\calW_1(\mu, \nu)$, then we have
\begin{align*}
\calR_{\rho}(\mu, \nu) &\leq \left(\sum_{i=1}^n \sum_{j=1}^n \mu_i \nu_j \left( \frac{\gamma_{ij}}{\mu_i \nu_j} \right)^{\rho} \|x_i - x_j\|_2^{\rho}\right)^{1/\rho} \leq \sup_{i, j} \left|\frac{1}{\mu_i\nu_j}\right|^{(\rho-1)/\rho} \left(\sum_{i=1}^n \sum_{j=1}^n \gamma_{ij}^{\rho} \|x_i - x_j \|_2^{\rho} \right)^{1/\rho} \\
	&\leq \sup_{i,j} \left| \frac{1}{\mu_i \nu_j} \right|^{(\rho-1)/\rho} \cdot \calW_{1}(\mu, \nu),
\end{align*}
where the last line uses the fact that the $\ell_{\rho}$-norm of a vector is always smaller than the $\ell_1$-norm of a vector.
\end{proof}}

\section{Lemmas and Proofs of Section~\ref{sec:sgd}}

\begin{lemma}\label{lem:preprocess}
For any $\rho \geq 1$ as well as parameters $\sigma, \sigma_{\mu}, \sigma_{\nu} \geq 0$, the distributions $\mu'$ and $\nu'$ satisfy:
\begin{itemize}
\item Every point in the support of $\mu'$ and every point in the support of $\nu'$ has distance between $\sigma r$ and $r \sqrt{1 + \sigma^2}$. 
\item Every element of the support of $\mu'$ is sampled with probability at least $\sigma_{\mu} / n$, and every element of $\nu'$ is sampled with probability at least $\sigma_{\nu} / m$. 
\item Both $\mu'$ and $\nu'$ are minor perturbations of $\mu$ and $\nu$, i.e.,
\[ \calR_{\rho}(\mu, \mu') \leq \left( n^{\rho-1} \cdot \frac{\sigma}{\sigma_{\mu}^{\rho-1}} + \sigma_{\mu} \right)^{1/\rho} \cdot r \qquad\text{and}\qquad \calR_{\rho}(\nu, \nu') \leq \sigma_{\nu}^{1/\rho} \cdot r. \]
\end{itemize} 
\end{lemma}
\begin{proof}
Consider the coupling $\gamma^{(1)} \in \R^{n \times n}$ between $\mu$ and $\mu'$ given by 
\[ \gamma_{ij}^{(1)} = \left\{\begin{array}{cc} \mu_i & i = j \text{ and } i \in [n] \setminus L_{\mu} \\
								\mu_i \mu_j / (1 - \zeta_{\mu}) & i \in L_{\mu} \text{ and } j \in [n] \setminus L_{\mu} \\
								0 & j \in L_{\mu} \end{array} \right. . \]
It is simple to verify, using the fact $\sum_{i \in L_{\mu}}\mu_i = \zeta_{\mu}$ that the above is indeed a coupling, and we now upper bound the cost:
\begin{align*}
\calR_{\rho}(\mu, \mu')^{\rho} &\leq \sum_{i \in [n] \setminus L_{\mu}} \frac{\mu_i^2}{1 - \zeta_{\mu}} \cdot \left(\frac{\mu_i(1-\zeta_{\mu})}{\mu_i^2} \cdot \sigma r \right)^{\rho} + \sum_{i \in L_{\mu}} \sum_{j \in [n] \setminus L_{\mu}} \frac{\mu_i \mu_j}{1 - \zeta_{\mu}} \cdot r^{\rho} \\
			&\leq \sum_{i \in [n] \setminus L_{\mu}} \mu_i \cdot (1/\mu_i)^{\rho-1} \cdot (\sigma r)^{\rho} + \zeta_{\mu} \cdot r^{\rho} \leq \left(n^{\rho-1} \cdot \frac{\sigma}{\sigma_{\mu}^{\rho-1}}  + \sigma_{\mu}\right) \cdot r^{\rho}.
\end{align*}
Similarly, we may write the coupling $\gamma^{(2)} \in \R^{m\times m}$ between $\nu$ and $\nu'$ given by
\[ \gamma_{ij}^{(2)} = \left\{ \begin{array}{cc} \nu_i & i = j \text{ and } i \in [m] \setminus L_{\nu} \\
									\nu_i \nu_j / (1 - \zeta_{\nu}) & i \in L_{\nu}\text{ and }j \in [m] \setminus L_{\nu} \\
									0 & j \in L_{\nu} \end{array} \right. ,\]
which allows us to upper bound $\calR_{\rho}(\nu, \nu')$ by
\begin{align*}
\calR_{\rho}(\nu, \nu')^{\rho} &\leq \sum_{i \in L_{\nu}} \sum_{j \in [m] \setminus L_{\nu}} \frac{\nu_i \nu_j}{1 - \zeta_{\nu}} \cdot r^{\rho} \leq \sigma_{\nu} \cdot r^{\rho}
\end{align*}
\end{proof}

\subsection{Additional Details from Section~\ref{sec:alg}}

\paragraph{Three Sub-routines}\label{sec:sub-routine-guarantees} Before stating the main lemma which we will prove for the guarantees on the above algorithm, we give the three lemmas which encapsulate the performance guarantees on $\EstAlpha$, $\EstBeta$, and $\EstPenalty$. Assuming these lemmas, we will then prove the main lemma, and show how that implies Theorem~\ref{thm:main}. We will defer the proof the three lemmas until after the analysis of the algorithm, as they rely on the data structures from Appendix~\ref{sec:data-structures}. Thus, the proofs of Lemma~\ref{lem:est-alpha} and Lemma~\ref{lem:est-beta} appear in Appendix~\ref{sec:est-alpha-proofs}, and the proof of Lemma~\ref{lem:penalty} appears in Appendix~\ref{sec:est-penalty-proof}.
\begin{lemma}[Guarantees on \EstAlpha]\label{lem:est-alpha}
Fix a parameter $s \geq 1$ and a small parameter $\tau > 0$, we also fix two distributions $\mu, \nu$ supported on points $\{x_1,\dots, x_n \}$ and $\{y_1,\dots, y_m\}$ whose pairwise distance is between $\sigma r$ and $r$. For any $(\alpha, \beta) \in \R^{n+m}$ and any $\eps > 0$, there is a randomized algorithm $\emph{\EstAlpha}$ with the following guarantees:
\begin{itemize}
\item The algorithm receives as input the vector $(\alpha, \beta) \in \R^{n+m}$, the accuracy parameters $\eps > 0$ and $\tau > 0$, and failure probability $\delta \in (0, 1)$. The algorithm produces as output a sequence of $n$ numbers $\boldeta_1,\dots, \boldeta_n \in \R_{\geq0}$. 
\item The algorithm runs in time $(n+m) \cdot \poly^*(2^s/ \eps)$ and with probability at least $1-\delta$, every $i \in [n]$ satisfies
\[ \left( s C_s \sum_{j=1}^m \nu_j \cdot \frac{((\alpha_i - \beta_j)^+)^{s-1}}{\|x_i - y_j\|_2^s} \right) - \tau \leq \boldeta_i \leq (1+\eps) \left( s C_s \sum_{j=1}^m \nu_j \cdot \frac{((\alpha_i - \beta_j)^+)^{s-1}}{\|x_i - y_j\|_2^s} \right). \]
\end{itemize}
\end{lemma}

\begin{lemma}[Guarantees on \EstBeta]\label{lem:est-beta}
Fix a parameter $s \geq 1$ and a small parameter $\tau > 0$, we also fix two distributions $\mu, \nu$ supported on points $\{x_1,\dots, x_n \}$ and $\{y_1,\dots, y_m\}$ whose pairwise distance is between $\sigma r$ and $r$. For any $(\alpha, \beta) \in \R^{n+m}$ and any $\eps > 0$, there is a randomized algorithm $\emph{\EstBeta}$ with the following guarantees:
\begin{itemize}
\item The algorithm receives as input the vector $(\alpha, \beta) \in \R^{n+m}$, the accuracy parameters $\eps > 0$ and $\tau > 0$, and failure probability $\delta \in (0, 1)$. The algorithm produces as output a sequence of $n$ numbers $\bxi_1,\dots, \bxi_m \in \R_{\geq0}$. 
\item The algorithm runs in time $(n+m) \cdot \poly^*(2^s / \eps)$ and with probability at least $1-\delta$, every $j \in [m]$ satisfies
\[ \left( sC_s \sum_{i=1}^n \mu_i \cdot \frac{((\alpha_i - \beta_j)^+)^{s-1}}{\|x_i - y_j\|_2^s} \right) - \tau \leq \bxi_j \leq (1+\eps) \left( s C_s \sum_{i=1}^n \mu_i \cdot \frac{((\alpha_i - \beta_j)^+)^{s-1}}{\|x_i - y_j\|_2^s} \right). \]
\end{itemize}
\end{lemma}

\begin{lemma}[Guarantees on \EstPenalty]\label{lem:penalty}
Fix a parameter $s \geq 1$ and a small parameter $\tau > 0$, we also fix two distributions $\mu, \nu$ supported on points $\{x_1,\dots, x_n \}$ and $\{y_1,\dots, y_m\}$ whose pairwise distance is between $\sigma r$ and $r$. For any $(\alpha, \beta) \in \R^{n+m}$ and any $\eps > 0$, there is a randomized algorithm $\emph{\EstPenalty}$ with the following guarantees:
\begin{itemize}
\item The algorithm receives as input the vector $(\alpha, \beta) \in \R^{n+m}$, the accuracy parameters $\eps > 0$ and $\tau > 0$, and failure probability $\delta \in (0, 1)$. The algorithm produces as output a number $\bomega \in \R_{\geq 0}$. 
\item The algorithm runs in time $(n+m) \cdot \poly^*(2^s/ \eps)$ and satisfies that with high probability,
\[ \left( C_s \sum_{i=1}^n \sum_{j=1}^m \mu_i \nu_j \left( \dfrac{(\alpha_i - \beta_j)^+}{\|x_i - y_j \|_2}\right)^s \right) - \tau \leq \bomega \leq (1+\eps) \left( C_s \sum_{i=1}^n \sum_{j=1}^m \mu_i \nu_j \left( \dfrac{(\alpha_i - \beta_j)^+}{\|x_i - y_j \|_2}\right)^s \right). \]
\end{itemize}
\end{lemma}

\subsection{Proofs from the Analysis of the Algorithm}

\begin{proof}[Proof of the Termination Condition]
First, we note that we can apply the translation so as to assume that $(\alpha_t, \beta_t)$ and $(\alpha^*, \beta^*)$ satisfies
\[ \| \alpha_t\|_{\infty}, \| \beta_t \|_{\infty}, \| \alpha^*\|_{\infty} , \|\beta^*\|_{\infty} \leq O(1) \cdot \left(\frac{nm}{\sigma_{\mu} \sigma_{\nu}}\right)^{(\rho-1)/\rho} \cdot r^{\rho}.\]
Let $\eta_1,\dots, \eta_n$ and $\xi_1,\dots, \xi_m$ denote the quantities
\[ \eta_i = sC_s \sum_{j=1}^m \nu_j \cdot \frac{((\alpha_i - \beta_j)^+)^{s-1}}{\|x_i - y_j \|_2^s} \qquad\text{and}\qquad \xi_j = sC_s \sum_{i=1}^n \mu_i \cdot \frac{((\alpha_i - \beta_j)^+)^{s-1}}{\|x_i - y_j\|_2^s}, \]
and consider the function
\[ \tilde{g}(\alpha, \beta) = \sum_{i=1}^n \mu_i\eta_i \alpha_i - \sum_{j=1}^m \nu_j\xi_j \beta_j - C_s \sum_{i=1}^n \sum_{j=1}^m \mu_i \nu_j \left(\frac{(\alpha_i - \beta_j)^+}{\|x_i - y_j\|_2} \right)^{s}. \]
Notice that $\tilde{g}$ is a concave function, and the partial derivatives at $(\alpha_t, \beta_t)$ are all zero, so that $\tilde{g}$ is maximized at $(\alpha_t, \beta_t)$. Thus,
\begin{align*}
g(\alpha_t, \beta_t) &= \tilde{g}(\alpha_t, \beta_t) + \sum_{i=1}^n \mu_i (1 - \eta_i) \alpha_i - \sum_{j=1}^m \nu_j (1 - \xi_j)\beta_j \\
			&\geq \tilde{g}(\alpha^*, \beta^*) + \sum_{i=1}^n \mu_i (1 - \eta_i) \alpha_i - \sum_{j=1}^m \nu_j (1 - \xi_j)\beta_j \\
			&= g(\alpha^*, \beta^*) + \sum_{i=1}^n \mu_i (1 - \eta_i) (\alpha_i - \alpha_i^*) - \sum_{j=1}^m \nu_j (1 - \xi_j)(\beta_j - \beta_j^*).
\end{align*}
This implies that 
\begin{align*}
g(\alpha^*, \beta^*) - g(\alpha_t, \beta_t) \leq (\| \alpha_t \|_{\infty} + \|\alpha^*\|_{\infty}) \sum_{i=1}^n \mu_i |1 - \eta_i| + (\|\beta_t\|_{\infty} + \|\beta^*\|_{\infty}) \sum_{j=1}^m \nu_j |1 - \xi_j|.
\end{align*}
By correctness of the algorithms $\EstAlpha$ and $\EstBeta$, $\boldeta_i$ is a good approximation to $\eta_i$ and $\bxi_j$ is a good approximation to $\xi_j$, which allows us to upper bound the above expression by
\begin{align*}
\left( \| \alpha_t\|_{\infty} + \| \alpha^*\|_{\infty} + \| \beta_t\|_{\infty} + \|\beta^*\|_{\infty} \right) \left( \eps_2 + \tau + \eps_1 \cdot s C_s \sum_{i=1}^n \sum_{j=1}^m \mu_i \nu_j \cdot \frac{((\alpha_i - \beta_j)^+)^{s-1}}{\|x_i - y_j\|_2^s}\right),
\end{align*}
which by Lemma~\ref{lem:bounded-penalty-grad} is at most
\[ O(1) \cdot \left(\frac{nm}{\sigma_{\mu} \sigma_{\nu}} \right)^{(\rho-1)/\rho} \cdot r^{\rho} \cdot \left( \eps_2 + \tau + \frac{\eps_1 s}{\sigma}\right),  \]
since $r / \rmin \leq \sigma$.\ignore{
We focus on the final term. We use the fact that the minimum distance between any set of points $x_i$ and $y_j$ is at least $\sigma r$, and convexity of $(\cdot)^{\rho}$
\begin{align*}
\sum_{i=1}^n \sum_{j=1}^m \mu_i \nu_j \cdot \frac{((\alpha_i - \beta_j)^+)^{s-1}}{\|x_i - y_j\|_2^s} &\leq \frac{1}{\sigma r} \sum_{i=1}^n \sum_{j=1}^m \mu_i \nu_j \left(\frac{(\alpha_i - \beta_j)^+}{\|x_i - y_j\|_2} \right)^{s-1} \\
&\leq \frac{1}{\sigma r} \left( \sum_{i=1}^n \sum_{j=1}^m \mu_i \nu_j \left(\frac{(\alpha_i - \beta_j)^+}{\|x_i - y_j\|_2} \right)^s \right)^{1/\rho} \leq \frac{4^{1/\rho}}{\sigma r \cdot C_s^{1/\rho} } \cdot \calR_{\rho}(\mu, \nu).
\end{align*}
where we used Lemma~\ref{lem:bounded-penalty} in the last inequality. Plugging this upper bound into the expression, and the upper bound on $\calR_{\rho}(\mu, \nu) \leq r \cdot (nm / (\sigma_{\mu} \sigma_{\nu}))^{(\rho-1)/\rho}$, from Lemma~\ref{lem:relate-w-1} and the granularity property of $\mu$ and $\nu$,
\begin{align*}
g(\alpha^*, \beta^*) - g(\alpha_t, \beta_t) &\lsim \left( \| \alpha_t\|_{\infty} + \| \alpha^*\|_{\infty} + \| \beta_t\|_{\infty} + \|\beta^*\|_{\infty} \right) \cdot \left(\eps_2 + \tau + \eps_1 s \cdot \frac{r}{\sigma r} \cdot \left( \frac{nm}{\sigma_{\mu} \sigma_{\nu}}\right)^{(\rho-1)/\rho} \right)\\
				&\leq O((\eps_1+\eps_2+\tau) s / \sigma) \cdot \left(\frac{nm}{\sigma_{\mu} \sigma_{\nu}} \right)^{2(\rho-1)} \cdot r^{\rho}. 
	\end{align*}}
\end{proof}

\begin{proof}[Proof that Updates Increase Objective (Lemma~\ref{lem:increase-obj})]
In order to simplify the notation, we will let $\beta = \beta_{t} = \beta_{t+1}$, and also denote $\tilde{\alpha} = \alpha_{t+1}$ and $\alpha = \alpha_{t}$. Then, we have
\begin{align*}
g(\tilde{\alpha}, \beta) - g(\alpha, \beta) &= \sum_{i=1}^n \mu_i (\tilde{\alpha}_i - \alpha_i) + C_s \sum_{i=1}^n \sum_{j=1}^m \frac{\mu_i \nu_j}{\|x_i - y_j\|_2^s} \cdot \left( ((\alpha_i - \beta_j)^+)^{s} - ((\tilde{\alpha}_i - \beta_j)^+)^{s}\right).
\end{align*}
We now use the following simple consequence of convexity of $((t)^+)^{s}$ and $((t)^+)^{s-1}$:
\begin{align*} 
&((\alpha_i - \beta_j)^+)^{s} - ((\tilde{\alpha}_i - \beta_j)^+)^{s} \\
&\qquad\qquad \geq s (\alpha_i - \tilde{\alpha}_i) \cdot ((\alpha_i - \beta_j)^+)^{s-1} - s(s-1) (\alpha_i - \tilde{\alpha}_i)^2 ((\alpha_i - \beta_j)^+)^{s-2}.
\end{align*}
In particular, we may lower bound $g(\tilde{\alpha}, \beta) - g(\alpha, \beta)$ by
\begin{align}
&\sum_{i=1}^n (\tilde{\alpha}_i - \alpha_i) \cdot  \mu_i \cdot \left( 1 - s C_s \sum_{j=1}^m \nu_j \cdot \frac{((\alpha_i-\beta_j)^+)^{s-1}}{\|x_i - y_j\|_2^s}\right) \label{eq:hah1}\\
&\qquad\qquad - s(s-1) C_s \sum_{i=1}^n (\tilde{\alpha}_i - \alpha_i)^2 \sum_{j=1}^m \mu_i \nu_j \cdot \frac{((\alpha_i - \beta_j)^+)^{s-2}}{\|x_i - y_j\|_2^s}. \label{eq:hah2}
\end{align}
We will lower bound the first term (\ref{eq:hah1}) and upper bound the second term (\ref{eq:hah2}). In particular, recall that due to our approximation guarantee on $\boldeta_i$, every $i \in [n]$ satisfies
\begin{align*}
\left| \boldeta_i - sC_s \sum_{j=1}^m \nu_j \cdot \frac{((\alpha_i - \beta_j)^+)^{s-1}}{\|x_i - y_j\|_2^s}\right| &\leq \tau + \eps_1 s C_s \sum_{j=1}^m \nu_j \cdot \frac{((\alpha_i - \beta_j)^+)^{s-1}}{\|x_i - y_j\|_2^s},
\end{align*}
and by our definition of the update, we may lower bound (\ref{eq:hah1}), using Lemma~\ref{lem:bounded-penalty-grad}, by
\begin{align*}
\lambda \sum_{i=1}^n \mu_i |1 - \boldeta_i| - \lambda \tau - \lambda \eps_1 \cdot s C_s \sum_{i=1}^n \sum_{j=1}^m \mu_i\nu_j \cdot \frac{((\alpha_i - \beta_j)^+)^{s-1}}{\|x_i - y_j\|_2^s} \geq \lambda \left( \eps_2 - \tau - \frac{4\eps_1 \cdot s}{\sigma} \right).
\end{align*}
Then, we may upper bound (\ref{eq:hah2}) by $(\tilde{\alpha}_i - \alpha_i)^2 \leq \lambda^2$ for every $i \in [n]$, which implies, also by Lemma~\ref{lem:bounded-penalty-grad} that (\ref{eq:hah2}) is at most
\begin{align*}
O(1) \cdot \frac{s (s-1) \cdot \lambda^2}{\sigma^2 \cdot r^{\rho}}.
\end{align*}
In particular, since $\lambda$ is a small constant factor of $\eps_2 \cdot r^{\rho} \cdot \sigma^2 /s^2$, and both $\tau$ and $\eps_1 s / \sigma$ are substantially smaller than $\eps_2$, we may lower bound
\begin{align*}
g(\tilde{\alpha}, \beta) - g(\alpha, \beta) \geq \Omega(\lambda \cdot \eps_2).
\end{align*}
\end{proof}
\ignore{\section{A Stochastic Gradient Ascent Algorithm}

\newcommand{\err}{\mathrm{err}}
\newcommand{\CheckStepSize}{\texttt{Check-Step-Size}}
\newcommand{\ApproxGradient}{\texttt{ApproxGradient}}

The algorithm will maintain a setting of the dual variables $(\alpha_t, \beta_t) \in \R^{n+m}$, an approximately optimal final dual variables $(\hat{\alpha}_t, \hat{\beta}_t) \in \R^{n+m}$, and a normalization quantity $Z_t \in \R_{\geq 0}$ which it will update in each iteration. We first describe the algorithm assuming access to two algorithms $\CheckStepSize$ and $\ApproxGradient$ which we specify later. We initialize $(\alpha_0, \beta_0) = (0, 0)$ and $(\hat{\alpha}_0, \hat{\beta}_0) = (0, 0)$. We start with a setting of 
\[ T = \frac{1}{4\eps} \qquad\text{and}\qquad \eta = (n+m) \cdot r^{\rho} \cdot \sup_{\substack{i \in [n] \\ j \in [m]}} \left| \frac{\min\{ \mu_i, \nu_j\}}{s C_s \cdot \mu_i \nu_j}\right|^{\rho-1}.\]
For $t =1, \dots, T$, the algorithm performs the following steps:
\begin{enumerate}
\item The algorithm executes $\CheckStepSize(\alpha_{t-1}, \beta_{t-1})$ and produces the output random variable $\blambda_{t} \in \R_{\geq 0}$, if $\eta \leq \eps r^{\rho} / \blambda_t$, we proceed. Otherwise, we update $\eta$ to $\eta / 2$ and $T = 2T$. We then update $t = 1$ and re-start the algorithm.
\item We execute the algorithm $\ApproxGradient(\alpha_{t-1}, \beta_{t-1})$ which returns a random vector $\bDelta_{t} \in \R^{n+m}$. 
\item We perform the following updates:
\begin{align*}
(\alpha_{t}, \beta_{t}) &\longleftarrow (\alpha_{t-1}, \beta_{t-1}) + \eta \cdot \bDelta_{t}. \\
(\hat{\alpha}_{t}, \hat{\beta}_{t}) &\longleftarrow \left(\frac{t}{t + 1}\right) \cdot (\hat{\alpha}_t, \hat{\beta}_t) + \frac{1}{t + 1} \cdot (\alpha_t, \beta_t).
\end{align*}
\end{enumerate}
If we have reached $t = T$, then we output
\[ \left( 1/ \rho \right)^{1/\rho} \left|  \sum_{i=1}^n \mu_i \cdot (\hat{\alpha}_T)_i - \sum_{j=1}^m \nu_j \cdot (\hat{\beta}_T)_j \right|^{1/\rho}.\]

For simplicity in the notation for this section (and in order to give a generic description of the algorithm), we will consider the distributions $\mu$ and $\nu$, supported on $\{x_1,\dots, x_n\}$ and $\{ y_1,\dots, y_m \}$ as fixed. We also may assume (without loss of generality) that $\{x_1,\dots, x_n \}$ and $\{y_1,\dots, y_m \}$ are disjoint by considering an infinitesimally small perturbation of the points. Then, we consider the parameter $\rho > 1$ and the H\"{o}lder conjugate $s \geq 1$, i.e., $1/\rho + 1/s = 1$, and write, for any $(\alpha, \beta) \in \R^{n+m}$, 
\[ F_{ij}(\alpha, \beta) = \left( \dfrac{(\alpha_i + \beta_j)^+}{\|x_i - y_j \|_2} \right)^{s} \qquad \text{and}\qquad f_{ij}(\alpha, \beta) = s \cdot \dfrac{((\alpha_i + \beta_j)^+)^{s-1}}{\|x_i - x_j\|_2^s}, \]
and the function
\[ g(\alpha, \beta) = \sum_{i=1}^n \mu_i \alpha_i + \sum_{j=1}^m \nu_j \beta_j - C_s \sum_{i=1}^n \sum_{j=1}^m \mu_i \nu_j \cdot F_{ij}(\alpha, \beta), \]
where $C_s = (1/s) (1-1/s)^{s-1}$.
\begin{lemma}[Convergence Guarantee]
Suppose that the algorithms $\emph{\ApproxGradient}(\alpha, \beta)$ and $\emph{\CheckStepSize}(\alpha, \beta)$ satisfy the following guarantees:
\begin{itemize}
\item The algorithm $\emph{\ApproxGradient}(\alpha, \beta)$ is a randomized algorithm which outputs a random vector $\bDelta \in \R^{n + m}$ whose expectation is the gradient $\nabla g(\alpha, \beta)$.
\item The randomized algorithm $\emph{\CheckStepSize}(\alpha, \beta)$ outputs a value $\blambda \in \R_{\geq 0}$ which satisfies
\[ \Ex_{\bDelta}\left[ \left\| \bDelta \right\|_2^2\right] \leq \blambda \]
with probability $1-1/\poly(nm)$ when $\bDelta$ is generated according to $\emph{\ApproxGradient}(\alpha, \beta)$.
\end{itemize}
Then, if $(\alpha^*, \beta^*)$ is the maximizer of $g$, whenever the algorithm reaches $t = T$ and produces an output, we satisfy 
\begin{align*}
\Ex\left[ g(\alpha^*, \beta^*) - g(\hat{\balpha}_t, \hat{\bbeta}_t) \right] \leq \left(\dfrac{(n + m) \cdot r^{\rho}}{8\cdot t \cdot \eta} \cdot \sup_{\substack{i \in [n] \\ j \in [m]}} \left| \frac{\min\{ \mu_i, \nu_j\}}{s C_s \cdot \mu_i\nu_j} \right|^{2(\rho-1)} + \eps\right)  r^{\rho}.
\end{align*}
\end{lemma}

\begin{proof}
We will follow the usual analysis of stochastic gradient descent. For ease in notation, we will let $z \in \R^{n+m}$ where the first $n$ coordinates index into $\alpha$ and the second $m$ into $\beta$; this way we may write $z_t = (\alpha_t, \beta_t)$ and $z^* = (\alpha^*, \beta^*)$. In particular, for any fixed execution of iterations $1,\dots, t-1$, we have
\begin{align*}
g(z^*) - g(z_{t-1}) &\leq \langle \nabla g(z_{t-1}) , z^* - z_{t-1} \rangle = \Ex_{\bDelta_t}\left[ \langle \bDelta_t, z^* - z_{t-1} \rangle \right] = \frac{1}{\eta} \cdot \Ex_{\bDelta_t}\left[ \langle \bz_t - z_{t-1}, z^* - z_{t-1}\rangle \right] \\
		&= \frac{1}{2 \cdot \eta} \cdot \left( \| z^* - z_{t-1} \|_2^2 - \Ex_{\bDelta_t}\left[\| z^* - \bz_{t}\|_2^2\right]\right)  + \frac{\eta}{2} \Ex_{\bDelta_t}\left[ \| \bDelta_t \|_2^2 \right],
\end{align*}
which implies that for any sequence $z_1,\dots, z_{t-1}$, we can write
\begin{align*} 
&g(z^*) - g\left( \frac{1}{t} \sum_{i=1}^t z_{i-1}\right) \leq \frac{1}{t} \sum_{i=1}^t (g(z^*) - g(z_{i-1})) \\
			&\qquad\qquad\leq \frac{1}{2t \cdot \eta} \sum_{i=1}^t \|z^*- z_{i-1}\|_2^2- \frac{1}{2t \cdot \eta} \sum_{i=1}^t \Ex_{\bDelta_i'} \left[\|z^* - \bz_i'\|_2^2 \right] + \frac{1}{2t} \sum_{i=1}^t \eta \cdot \Ex_{\bDelta_i'}\left[ \| \bDelta_i'\|_2^2\right].
\end{align*}
which means that if we always have $\eta \leq 2\eps \cdot r^{\rho} / \blambda_i \leq 2\eps \cdot r^{\rho} / \Ex[\|\bDelta_i'\|_2^2]$, and the guarantees of $\CheckStepSize(\alpha_t, \beta_t)$ hold, the above quantity is bounded by
\begin{align*}
g(z^*) - g\left( \frac{1}{t} \sum_{i=1}^t z_{i-1}\right)  &\leq \frac{1}{2t \cdot \eta} \sum_{i=1}^t \|z^*- z_{i-1}\|_2^2- \frac{1}{2t \cdot \eta} \sum_{i=1}^t \Ex_{\bDelta_i'} \left[\|z^* - \bz_i'\|_2^2 \right] + \eps \cdot r^{\rho},
\end{align*}
which implies by linearity of expectation and the telescoping sum, that 
\begin{align*}
\Ex_{\bDelta_1,\dots, \bDelta_t}\left[ g(z^*) - g\left( \frac{1}{t} \sum_{i=1}^t z_{i-1}\right)\right] \leq \frac{\| z^*\|_2^2}{2t \cdot \eta} + \eps \cdot r^{\rho}.
\end{align*}
The final bound follows from the upper bound on $\| (\alpha^*, \beta^*)\|_{2}^2 \leq (n+m) \|(\alpha^*, \beta^*)\|_{\infty}^2$.
\end{proof}

\newcommand{\out}{\mathtt{out}}

\begin{lemma}[Output Guarantee]
Suppose $(\hat{\alpha}, \hat{\beta}) \in \R^{n+m}$ is the final point in the algorithm, which satisfies the following three guarantees for a small parameter $\eps \in (0, 1/4)$,
\begin{itemize}
\item First, we have $g(\alpha^*, \beta^*) - g(\hat{\alpha}, \hat{\beta}) \leq \eps \cdot r^{\rho}$.
\item Second, we have $\| (\hat{\alpha}, \hat{\beta})\|_2 \cdot \| \nabla g(\hat{\alpha}, \hat{\beta}) \|_2 \leq s \cdot \eps r^{\rho}$. 
\end{itemize}
Then, if we let $\out \in \R_{\geq 0}$ be defined as
\[ \out = \left(1/\rho \right)^{1/\rho} \left| \sum_{i=1}^n \mu_i \hat{\alpha}_i - \sum_{j=1}^m \nu_j \hat{\beta}_j\right|^{1/\rho}, \]
we have
\begin{align*} 
&\left| \out - \calR_{\rho}(\mu, \nu)\right| \leq O(\eps r / \rho).
\end{align*}
\end{lemma}

\begin{proof}
First, note that $\calR_{\rho}(\mu, \nu)^{\rho} = g(\alpha^*, \beta^*)$. Then, we note that we may re-write
\begin{align*}
g(\hat{\alpha}, \hat{\beta}) &= \sum_{i=1}^n \mu_i \hat{\alpha}_i - \sum_{j=1}^m \nu_j \hat{\beta}_j \\
	&\qquad \qquad - \frac{1}{s} \sum_{i=1}^n \mu_i \hat{\alpha}_i \cdot s C_s \sum_{j=1}^m \nu_j \cdot \frac{((\hat{\alpha}_i - \hat{\beta}_j)^+)^{s-1}}{\|x_i - y_j\|_2^s}\\
	&\qquad\qquad + \frac{1}{s} \sum_{j=1}^m \nu_j \hat{\beta}_j \cdot s C_s \sum_{i=1}^n \mu_i \cdot \frac{((\hat{\alpha}_i - \hat{\beta}_j)^+)^{s-1}}{\|x_i - y_j\|_2^s}\\
	&= \left(1 - \frac{1}{s} \right) \left( \sum_{i=1}^n \mu_i \hat{\alpha}_i - \sum_{j=1}^m \nu_j \hat{\beta}_j \right) + \frac{1}{s} \sum_{i=1}^n \hat{\alpha}_i \cdot \nabla g(\hat{\alpha}, \hat{\beta})_i - \frac{1}{s} \sum_{j=1}^m \hat{\beta}_j \cdot \nabla g(\hat{\alpha},\hat{\beta})_j \\
	&= \out^{\rho} + \frac{1}{s} \sum_{i=1}^n \hat{\alpha}_i \nabla g(\hat{\alpha}, \hat{\beta}) - \frac{1}{s} \sum_{j=1}^m \hat{\beta}_j \nabla g(\hat{\alpha}, \hat{\beta})_j
\end{align*} 
By the Cauchy-Schwarz inequality, and the fact $g(\hat{\alpha}, \hat{\beta})$ and $g(\alpha^*, \beta^*)$ differ by at most $\eps r^{\rho}$, we have
\begin{align*}
\left|\calR_{\rho}(\mu, \nu)^{\rho} - \out^{\rho} \right| \leq 2\eps \cdot r^{\rho} 
\end{align*}
Thus, we conclude by realizing that
\begin{align*}
\left| \frac{\calR_{\rho}(\mu, \nu) - \out}{r} \right| &= \left| \left(\left(\frac{\calR_{\rho}(\mu, \nu)}{r}\right)^{\rho}\right)^{1/\rho} - \left(\left(\frac{\out}{r} \right)^{\rho}\right)^{1/\rho} \right| \\
	&\leq \left| \left(\frac{\calR_{\rho}(\mu, \nu)}{r}\right)^{\rho} - \left(\frac{\out}{r}\right)^{\rho} \right| \cdot \frac{1}{\rho} \cdot \min\left\{ \left(\frac{\calR_{\rho}(\mu, \nu)}{r}\right)^{\rho}, \left(\frac{\out}{r}\right)^{\rho}\right\}^{1/\rho - 1} \\
	&\leq \frac{2 \eps}{\rho} \left( 1 - \frac{1}{r^{\rho}}\left| \calR_{\rho}(\mu,\nu)^{\rho} - \out^{\rho} \right| \right)^{1/\rho - 1} \leq \frac{2\eps}{\rho} \cdot \left( 1 - 2\eps\right)^{1/\rho - 1}.
\end{align*}
\end{proof}}

\section{Augmenting Kernel Density Estimation Data Structures}\label{sec:data-structures}

This section gives the algorithms for $\EstAlpha$, $\EstBeta$, and $\EstPenalty$ giving the proofs of Lemma~\ref{lem:est-alpha}, Lemma~\ref{lem:est-beta}, and Lemma~\ref{lem:penalty}. We first draw the connection to kernel density estimation and define the modified data structure problem that we will need. Then, Lemma~\ref{lem:est-alpha}, Lemma~\ref{lem:est-beta} and Lemma~\ref{lem:penalty} will follow from different instantiations of one data structure.

\newcommand{\KDE}{\texttt{KDE}}
\newcommand{\AugmentedKDE}{\texttt{Augmented-KDE}}

\begin{definition}[Kernel Density Estimation]
Let $\sfK \colon \R^d \times \R^d \to \R_{\geq 0}$ be a function, $\Phi > 1$ an aspect ratio bound, $\eps > 0$ a multiplicative error parameter, and $\delta > 0$ a failure probability. $\emph{\KDE}(\sfK, \Phi, \eps, \delta)$ is the following data structure problem:
\begin{itemize}
\item \emph{\textbf{Preprocessing}}: The data structure receives a set of points $X = \{ x_1,\dots, x_n \} \subset \R^d$.
\item \emph{\textbf{Query}}: A query is specified by a point $y \in \R^d$, and we will have the promise that $\max_{i} \|x_i - y\|_2 / \min_i \|x_i - y\|_2$ is at most $\Phi$. The data structure should output an estimate $\hat{\bxi} \in \R_{\geq 0}$. 
\end{itemize}
The guarantee is that for any dataset and any query $y \in \R^d$, with probability at least $1 - \delta$ over the randomness in constructing the data structure,
\begin{align*}
(1-\eps) \sum_{i =1}^n \sfK(x_i, y) \leq \hat{\bxi} \leq (1+\eps) \sum_{i=1}^n \sfK(x_i, y).
\end{align*}
\end{definition}

In using data structures for kernel density estimation, we will instantiate the data structure for sets of vectors which will be subsets of the support of the distributions $\mu$ and $\nu$. In addition, the aspect ratio bound will be $\Phi = 1/\sigma$ (since we consider inputs whose distance is at most $r$ and the minimum distance is at least $\sigma r$). For a small parameter $\eps_0 > 0$, we will be interested in kernel functions $\sfK \colon \R^d \times \R^d \to \R_{\geq 0}$ of the form:
\begin{align} 
\sfK(x_i, y) &= \dfrac{1}{\eps_0 \cdot (\sigma r)^{s} + \|x_i - y\|_2^s}. \label{eq:kernel}
\end{align}
The kernel (\ref{eq:kernel}) is a scaled Student-$t$ Kernel, and falls within the kernels explored  in~\cite{BCIS18}. The results of \cite{BCIS18} hold more generally for classes of ``smooth'' kernels, where they formally define $(L, t)$-smooth kernels (see Definition~1 in \cite{BCIS18}). We note that $\sfK$ in (\ref{eq:kernel}) is a $(1, s)$-smooth kernel, so that their results will apply with $L = 1$ and $t = s$. For this setting, we have every $i \in [n]$ and $y \in \R^d$ within distance between $\sigma r$ and $r$ from $x_i$, 
\begin{align*}
\frac{1-\eps_0}{\|x_i - y \|_2^s} \leq \sfK(x_i, y) \leq \frac{1}{\|x_i - y\|_{2}^s} .
\end{align*}

\newcommand{\UpdateAdd}{\texttt{UpdateAdd}}
\newcommand{\UpdateRem}{\texttt{UpdateRem}}
\newcommand{\UpdateAddA}{\texttt{UpdateAddA}}
\newcommand{\UpdateRemA}{\texttt{UpdateRemA}}
\newcommand{\Preprocess}{\texttt{Preprocess}}
\newcommand{\Query}{\texttt{Query}}

\newcommand{\PreprocessA}{\texttt{PreprocessA}}
\newcommand{\QueryA}{\texttt{QueryA}}

\begin{theorem}[Main Theorem of \cite{BCIS18}, instantiated to $\sfK$ in (\ref{eq:kernel})]\label{thm:bcis}
For any $\Phi > 1$, and $\eps, \delta > 0$, there exists two randomized algorithms $\emph{\Preprocess}$, and $\emph{\Query}$ for solving $\emph{\KDE}(\sfK, \Phi, \eps, \delta)$, with the following guarantees:
\begin{itemize}
\item $\emph{\Preprocess}(X)$ receives as a dataset $X = \{x_1,\dots, x_n \} \subset \R^d$, and outputs a pointer $v$ to a data structure for $\emph{\KDE}(\sfK, \Phi, \eps, \delta)$.
\item $\emph{\Query}(v, y)$ receives as input a pointer to a data structure $v$ for $\emph{\KDE}(\sfK, \Phi, \eps, \delta)$ and returns the query at $y$ for $\emph{\KDE}(\sfK, \Phi, \eps, \delta)$.
\end{itemize} 
We are guaranteed that $\emph{\Query}$ takes time $\poly^*(2^{s} / \eps)$, and the algorithm $\emph{\Preprocess}$ takes time $O(n) \cdot \poly^*(2^s/\eps)$.
\end{theorem}
\ignore{
\begin{remark}
We note that \cite{BCIS18} state their results for a static version of the problem (i.e., the set $\Omega = \{x_1,\dots, x_n \}$ is presented up front and not modified). They show that a data structure which preprocesses a set $\Omega = \{x_1,\dots, x_n \} \subset \R^d$ time $n \cdot \poly(d,\log n, \log \Phi, 1/\eps, \log(1/\delta), 2^{s})$, and it is not too hard to see that their construction readily implies dynamic algorithms which process each update in time $\poly(d, \log n, \log \Phi, 1/\eps, \log(1/\delta), 2^{s})$. 

In addition, it makes the analysis cleaner if we view $\emph{\Query}(v, y)$ as being deterministic (even though $\emph{\UpdateAdd}$ and $\emph{\UpdateRem}$ is randomized). It's not hard to see by inspection of \cite{BCIS18} that this modification is without loss of generality. Whenever the data structure is initialized or makes an update, it pre-samples the randomness used for future queries, and the randomness is stored.
\end{remark}}

We now introduce the augmented data structure problem which we need in order to solve $\EstAlpha$, $\EstBeta$, and $\EstPenalty$. The goal is to incorporate the fact that points have some associated real values $\alpha, \beta$.

\begin{definition}[Augmented Kernel Density Estimation]
Let $s_2 \geq 1$ be a parameter, $\Phi > 1$ is an aspect ratio bound, $\eps > 0$ be a multiplicative error parameter, and $\delta > 0$ be a desired failure probability. $\emph{\AugmentedKDE}(\sfK, s_2, \Phi, \eps, \delta)$ is the following data structure problem.
\begin{itemize}
\item \emph{\textbf{Preprocessing}}: We receive a set of points $X = \{ x_1,\dots, x_n \} \in \R^d$. In addition, each point has an associated weight $\alpha_i \in \R$ with $|\alpha_i| \leq r \cdot \poly(dn\Phi 2^s/\eps)$ and a parameter $\mu_i \in [1/\poly(n), 1]$, for a parameter $r \geq 0$ which will be the maximum distance considered.
\ignore{\item \emph{\textbf{Updates}}: An update adds or removes a point $x_i$ from $\Omega$ such that the aspect ratio is still at most $\Phi$, as well as updates its associated weight. }
\item \emph{\textbf{Query}}: A query is specified by a point $y \in \R^d$ and weight $\beta \in \R$. We are promised that:
\begin{itemize}
\item The point $y \in \R^d$ satisfies $\max_{i \in [n]} \|x_i - y\|_2 \leq r$ and that $\min_i \| x_i - y\|_2$ is at least $\sigma r$.
\item In addition, for every $i$, $|\alpha_i - \beta| \in \{0 \} \cup [ \sigma r / \poly(dn\Phi 2^s/\eps), r \cdot \poly(dn\Phi 2^s/\eps)]$, and the data structure outputs a quantity $\hat{\boldeta} \in \R_{\geq 0}$. 
\end{itemize}
\end{itemize}
The guarantee is that for any fixed query, with probability at least $1 - \delta$ over the randomness in the construction of the data structure,
\begin{align*}
(1-\eps) \sum_{i=1}^n \mu_i \cdot ((\alpha_{i} - \beta)^+)^{s_2} \cdot \sfK(x_i, y) \leq  \hat{\bxi} \leq (1+\eps)  \sum_{i = 1}^n \mu_i \cdot ((\alpha_{i} - \beta)^+)^{s_2} \cdot \sfK(x_i, y).
\end{align*} 
\end{definition}

\begin{theorem}\label{thm:aug-kde}
For any $s_2 \geq 1$, $\Phi > 1$, and $\eps, \delta > 0$, there exists three randomized algorithms $\emph{\PreprocessA}$ and $\emph{\QueryA}$ for solving $\emph{\AugmentedKDE}(\sfK, s_2, \Phi, \eps, \delta)$, with the following guarantees:
\begin{itemize}
\item $\emph{\PreprocessA}(X, \alpha)$ receives as input a dataset $X$ of at most $n$ points, and a vector $\alpha$ indicating a weight for each point and the vector $\mu$. The algorithm outputs a pointer $v$ to a data structure for $\emph{\AugmentedKDE}(\sfK, s_2, \Phi, \eps, \delta)$. \ignore{ a point $x \in \R^d$, and its associated weight $\alpha \in \R$. The algorithm modifies the data structure stored at $v$ to maintain $\Omega \leftarrow \Omega \cup \{ x \}$ with its associated weight $\alpha$ with $|\alpha|\leq r \cdot \poly(nd\Phi/\eps)$.}
\ignore{\item $\emph{\UpdateRemA}(v, x, \alpha)$ receives as input a pointer to a data structure for $\emph{\AugmentedKDE}(\sfK, s_2, \Phi, \eps, \delta)$, a point $x \in \R^d$, and its associated weight $\alpha \in \R$. The algorithm is promised that $x \in \Omega$ with weight $\alpha$, and it modifies the data structure stored at $v$ to maintain $\Omega \leftarrow \Omega \setminus \{ x \}$ (and deletes its associated weight $\alpha$).}
\item $\emph{\QueryA}(v, y, \beta)$ receives as input a pointer to a data structure for $\emph{\AugmentedKDE}(\sfK, s_2, \Phi, \eps, \delta)$, a point $y \in \R^d$, and a weight $\beta \in \R$ such that $|\alpha_i  - \beta| \in \{0 \} \cup \{ \sigma r / \poly(dn\Phi 2^s/\eps), r \cdot \poly(dn\Phi2^s/\eps)]$. The algorithm outputs query at $y$ with weight $\beta$ for $\emph{\AugmentedKDE}(\sfK, s_2, \Phi, \eps, \delta)$.
\end{itemize}
We are guaranteed that $\emph{\QueryA}$ takes time $\poly^*(2^{s+s_2} / \eps)$, and $\emph{\PreprocessA}$ takes time $O(n) \cdot \poly^*(2^{s+s_2} / \eps)$.
\end{theorem}

\subsection{Proof of Theorem~\ref{thm:aug-kde}}\label{sec:proof-aug-kde}

Since we are promised that every index $\mu_1,\dots, \mu_n$ is a number between $\mu_{\min} = 1/\poly(n)$ and $1$, and we can output a multiplicative $1\pm \eps$-approximation to the final sum, we will partition the set of points into $O(\log n / \eps)$ many parts, according to the range for which $\mu_i \in [\mu_{\min} (1+\eps)^j, \mu_{\min}(1+\eps)^{j+1}]$. Then, it suffices to output, for each of the $O(\log n /\eps)$ many ranges $j$, a $1\pm \eps$-approximation to the quantity
\[ \sum_{i \in P_j} ((\alpha_i - \beta)^+)^{s_2} \cdot \sfK(x_i, y). \]
For the remainder of the discussion we will assume that we have performed this partition (to drop $\mu$ from the notation), and assume henceforth that all of the weights specified by $\mu$ are equal. 

We refer to Figure~\ref{fig:data-structure-format} for the description of the data structure, which maintains a binary tree over the points in $X$ sorted according to $\alpha$, where each internal node of the tree additionally holds a pointer to a $\KDE(\sfK, \Phi, \eps, \delta)$ data structure. From the description of Figure~\ref{fig:data-structure-format}, the algorithms $\PreprocessA$ is straight-forward, and we will mostly give and analyze $\QueryA$.

\newcommand{\ds}{\mathrm{ds}}
\newcommand{\leftchild}{\mathrm{LeftChild}}
\newcommand{\rightchild}{\mathrm{RightChild}}
\newcommand{\med}{\mathrm{med}}

\begin{figure}
	\begin{framed}
		\noindent Data Structure for $\AugmentedKDE(\sfK, s_2, \Phi, \eps, \delta)$
		\begin{flushleft}
			\noindent {\bf Preprocessing:}  The data structure preprocesses a set $X = \{ x_1,\dots, x_n \} \subset \R^d$, where each point has its associated weight $\alpha_1,\dots, \alpha_n \in \R$ and a weight $\mu_1,\dots, \mu_n$ (which after a partitioning step, we will assume are equal -- see Subsection~\ref{sec:proof-aug-kde}). \\ 
			\noindent{\bf Pointer:} $v$ will be the pointer to the root of a binary tree. 
		\begin{itemize}
				\item The data structure is organized into a balanced binary tree of depth $O(\log n)$, where the $n$ leaves correspond to the points of $X$ stored in sorted order according to their weights $\alpha_1,\dots, \alpha_n$.
				\item Each node $v$ of the binary tree maintains the following information:
				\begin{itemize}
					\item A set $v.S \subset X$ in the subtree of $v$.
					\item A pointer $v.\ds$ to a data structure for $\KDE(\sfK, \Phi, \eps, \delta\eps^2 / (O(n) \cdot 2^{O(s_2)} )$ storing $v.S$,
					\item Three numbers $v.\min ,v.\max \in \R$ such that 
					\begin{align*} 
					v.\min &= \min\left\{ \alpha_i : x_i \in v.S \right\}, \\
					v.\max &= \max\left\{ \alpha_i : x_i \in v.S \right\}, \\
					v.\med &= \median\left\{ \alpha_i : x_i \in v.S \right\} 
					\end{align*}
					\item If $v.S$ contains more than one point, it has two children $v.\leftchild$ and $v.\rightchild$. The left child $v.\leftchild$ stores the points $x_i \in v.S$ where $\alpha_i \leq v.\med$ and the right child $v.\rightchild$ stores the points $x_i \in v.S$ where $\alpha_i > v.\med$.
				\end{itemize}
				\item The algorithms $\PreprocessA(X, \alpha)$ works by first building the  balanced tree, and in the sorted order of $\alpha$. Furthermore, for every internal node $v$ we consider the dataset $v.S$ and execute $\Preprocess(v.S)$ and store the data structure in $v.\ds$.
			\end{itemize}
		\end{flushleft}
	\end{framed}
	\caption{Data Structure for $\AugmentedKDE$.} \label{fig:data-structure-format}
\end{figure}

\begin{figure}
	\begin{framed}
		\noindent Algorithm $\QueryA(v, y, \beta)$
		\begin{flushleft}
			\noindent {\bf Input:}  A pointer to a data structure $v$ for $\AugmentedKDE(\sfK, s_2, \Phi, \eps, \delta)$, a point $y \in \R^d$, and a weight $\beta \in \R$. \\ 
			\noindent{\bf Output:} An estimate $\hat{\bxi} \in \R_{\geq 0}$. 
		\begin{enumerate}
				\item We first check whether $v.\max \leq \beta$. If so, then every weight $\alpha_i - \beta \leq 0$ and hence $(\alpha_i - \beta)^+ = 0$, so output $\hat{\bxi} = 0$.
				\item Otherwise, let $k = \left\lceil \log_2 \left(\frac{(v.\max - \beta)}{\eps_0 \sigma r} \cdot \poly(n d \Phi 2^s/\eps)\right)\right\rceil$ (which will become ``hidden'' in the notation $\poly^*(\cdot)$), and consider the $k+2$ indices $\sigma_0, \dots, \sigma_k \in [0, v.\max - \beta]$ where
				\[ \sigma_{\ell} = \left\{ \begin{array}{cc} 0 & \ell = 0 \\
										 \left(\frac{\eps_0 \sigma r}{\poly(nd\Phi2^s/\eps)}\right) \cdot 2^{\ell-1} & \ell > 0 \\
										  v.\max - \beta & \ell = k + 1 \end{array} \right. , \]
				and let $I_0, \dots, I_{k}$ be the disjoint and consecutive intervals $I_\ell = (\beta + \sigma_{\ell}, \beta + \sigma_{\ell+1}]$ which partition $(\beta, v.\max]$.
				\item For each $\ell \in \{ 1, \dots, k\}$, and $t \in [T]$, for $T = 2^{O(s_2)} / \eps^2$, we perform the following:
				\begin{itemize}
					\item Sample $\bw_{\ell, t} \sim [\sigma_{\ell}^{s_2}, \sigma_{\ell+1}^{s_2}]$ uniformly at random.
					\item Let $\calV_{\ell}$ be the set of all nodes $u$ where $[u.\min, u.\max] \subset I_{\ell}$, and let $\calbV_{\ell} (\bw_{\ell, t}) = \{ \bv^{(1)},\dots, \bv^{(h)}\}$ be the minimal subset of $\calV_{\ell}$ which satisfies
					\[ (\bv^{(1)}.S, \bv^{(2)}.S, \dots , \bv^{(h)}.S) \text{ partition } \left\{ x_i \in \Omega : \alpha_i \in I_{\ell} \text{ and } \alpha_i \geq \beta + \bw_{\ell, t}^{1/s_2} \right\},\]
					and note that $h = O(\log n)$, and we may identify these nodes in $O(\log n)$ time.
					\item For each $l \in [h]$, we execute $\Query(\bv^{(l)}.\ds, y)$ and let $\hat{\bzeta}_{\ell, t, l}$ be its output, and let
					\[ \hat{\bxi}_{\ell, t} = \sum_{l=1}^h \hat{\bzeta}_{\ell, t, l}. \]
				\end{itemize}
				\item We output 
				\[ \hat{\bxi} = \dfrac{1}{T} \sum_{\ell=1}^{k} \sum_{t =1}^{T} (\sigma_{\ell+1}^{s_2} - \sigma_{\ell}^{s_2}) \cdot \hat{\bxi}_{\ell, t}. \]
			\end{enumerate}
		\end{flushleft}
	\end{framed}
	\caption{Description for $\QueryA$ Algorithm.} \label{fig:query-A}
\end{figure}

\begin{lemma}
An execution of $\emph{\QueryA}(v, y, \beta)$ takes time $\poly^*(2^{s+s_2} / \eps)$. 
\end{lemma}

\begin{proof}
The above claim on the running time of $\QueryA(v, y, \beta)$ follows from inspection of Figure~\ref{fig:query-A}. Indeed, Step 1 takes $O(1)$ time and in Steps~2 to 3, there are $\poly^*(2^{s_2} / \eps)$ many calls to $\Query$, where each takes time $\poly^*(2^{s}/\eps)$, by Theorem~\ref{thm:bcis}. 
\end{proof}

\begin{definition}
Consider a construction of the data structure for $\emph{\AugmentedKDE}(\sfK, s_2, \Phi,\eps,\delta)$, and suppose we fix the randomness and let $\calV$ be the set of all nodes in the tree rooted at $v$. For $y \in \R^d$, we consider the collection $\{ \hat{\bzeta}_{u, t}(y) : u \in \calV \}$, where $\hat{\bzeta}_{u}(y)$ is the output of $\emph{\Query}(u.\ds, y)$ (since we assumed $\emph{\Query}(u.\ds, y)$, is deterministic, we don't require adding the additional parameters $t \in [T]$ in case it is called multiple times).
\end{definition}

\begin{lemma}\label{cl:kde-correct}
With probability at least $1 - \delta/2$ over the construction of the data structures for  $\emph{\KDE}(\sfK, \Phi, \eps, \delta)$, every node $u \in \calV$ satisfies
\[(1-\eps) \sum_{x \in u.S} \sfK(x, y) \leq \hat{\bzeta}_u(y) \leq (1+\eps) \sum_{x \in u.S} \sfK(x, y).\]
\end{lemma}

\begin{proof}
We apply Theorem~\ref{thm:bcis}, and the fact that, in Figure~\ref{fig:data-structure-format}, we've instantiated the data structures with failure probability $\delta\eps^2 / (O(n) \cdot 2^{O(s_2)})$, such that we can union bound over all nodes in $\calV$.
\end{proof}

The remainder of the argument proceeds by computing the expectation $\hat{\bxi}$ over the randomness in $\{ \bxi_{\ell, t} : \ell \in [h], t \in [T] \}$ as well as the variance. Specifically, we show that $\Ex[\hat{\bxi}]$ satisfies the output guarantees, and that $\Var[\hat{\bxi}] \leq \eps \Ex[\hat{\bxi}]^2$, such that we can apply Chebyshev's inequality. Establishing the correctness of the estimate with high probability follows from a standard repetition argument.

\begin{lemma}
Consider a fixed construction of the data structure, and suppose that the conclusion of Claim~\ref{cl:kde-correct} holds (which occurs with probability at least $1-\delta/2$). Then, for any fixed query $y \in \R^d$ with weight $\beta$, the expectation of $\hat{\bxi}$ over the randomness in $\emph{\QueryA}(v, y, \beta)$,
\begin{align*}
(1-\eps) \sum_{i=1}^n ((\alpha_i - \beta)^+)^{s_2} \cdot \sfK(x_i, y)  \leq \Ex\left[ \hat{\bxi} \right] \leq (1+\eps) \sum_{i=1}^n ((\alpha_i - \beta)^+)^{s_2} \cdot \sfK(x_i, y).
\end{align*}
\end{lemma}

\begin{proof}
We have that for any $\ell \in \{1,\dots, k \}$ and $t \in [T]$,
\begin{align}
\hat{\bxi}_{\ell, t} = \sum_{u \in \calV_{\ell}} \ind\{ u \in \calV_{\ell}(\bw_{\ell, t}) \} \cdot \hat{\bzeta}_{u}(y) &\leq (1+\eps) \sum_{u \in \calV_{\ell}} \sum_{x \in u.S} \ind\{ u \in \calbV_{\ell}(\bw_{\ell, t}) \} \cdot \sfK(x, y) \nonumber \\
	   &= (1+\eps) \sum_{x \in X} \ind\{ \exists u \in \calbV_{\ell}(\bw_{\ell, t}), x \in u.S\} \cdot \sfK(x, y), \label{eq:exp-ub}
\end{align}
where in the second line, we used that every $x \in X$ which appears in some $u \in \calbV_{\ell}(\bw_{\ell, t})$ appears at most once. Similarly, 
\begin{align}
\hat{\bxi}_{\ell, t} \geq (1-\eps) \sum_{x \in X} \ind\{ \exists u \in \calbV_{\ell}(\bw_{\ell, t}), x \in u.S\} \cdot \sfK(x, y). \label{eq:exp-lb}
\end{align}
Then, for every $x_i \in X$ with weight $\alpha_i$,
\begin{align}
\Prx_{\bw_{\ell, t} \sim [\sigma_{\ell}^{s_2}, \sigma_{\ell+1}^{s_2}]}\left[ \exists u \in \calbV_{\ell}(\bw_{\ell, t}), x_i \in u.S\right] &= \Prx_{\bw_{\ell, t}}\left[ \beta + \bw_{\ell, t}^{1/s_2} \leq \alpha_i \leq \beta + \sigma_{\ell+1}\right] \nonumber\\
	&= \ind\{ \alpha_i \in I_{\ell}\} \cdot \dfrac{((\alpha_i - \beta)^+)^{s_2}}{\sigma_{\ell+1}^{s_2} - \sigma_{\ell}^{s_2}}.\label{eq:prob-include}
\end{align}
In particular, we may upper and lower bound the expectation of $\hat{\bxi}$ over the randomness in drawing $\bw_{\ell, t}$ by plugging (\ref{eq:prob-include}) into (\ref{eq:exp-ub}) and (\ref{eq:exp-lb}). Namely, we first note that for a fixed $\ell \in \{1,\dots, k \}$, all the draws from $t \in [T]$ of $\bw_{\ell, t}$ are identically distributed, so we may simplify
\begin{align*}
\Ex\left[ \hat{\bxi} \right] &= \sum_{\ell=1}^{k} \left( \sigma_{\ell+1}^{s_2} - \sigma_{\ell}^{s_2} \right)\Ex_{\bw_{\ell, t}}\left[ \hat{\bxi}_{\ell, t}\right].
\end{align*}
Then, we have that (\ref{eq:prob-include}) and (\ref{eq:exp-ub}) implies
\begin{align*}
 \left( \sigma_{\ell+1}^{s_2} - \sigma_{\ell}^{s_2} \right) \Ex_{\bw_{\ell, t}}\left[ \hat{\bxi}_{\ell, t}\right] \leq (1+\eps) \sum_{i=1}^{n}\ind\{ \alpha_i \in I_{\ell} \} \cdot ((\alpha_i - \beta)^+)^{s_2} \cdot \sfK(x_i, y).
\end{align*}
The lower bound proceeds similarly, expect we plug (\ref{eq:prob-include}) into (\ref{eq:exp-lb}). In particular, since there is no $i \in [n]$ where $\alpha_i \in I_0$, once $\ell \in \{1,\dots, k\}$, we cover the entire interval $(\beta, v.\max]$. 
\end{proof}

\begin{lemma}
Consider a fixed construction of the data structure, and suppose that the conclusion of Claim~\ref{cl:kde-correct} holds (which occurs with probability at least $1-\delta/2$). Then, for any fixed query $y \in \R^d$ with weight $\beta$, the variance of $\hat{\bx}$ over the randomness in $\emph{\QueryA}(v, y, \beta)$ satisfies
\begin{align*} 
\Var\left[ \hat{\bxi} \right] \leq \eps \left(\Ex\left[ \hat{\bxi}\right] \right)^2.
\end{align*}
\end{lemma}

\begin{proof}
For various settings of $\ell \in \{1, \dots, k\}$ and $t \in [T]$, the draws of $\bw_{\ell, t}$ are independent, so that we may write
\begin{align*}
\Var\left[ \hat{\bxi}  \right] = \sum_{\ell=1}^k \frac{1}{T} \cdot \Var\left[ \hat{\bxi}_{\ell, t}\right],
\end{align*}
and it suffices to upper bound that. We note that using the same upper bound in (\ref{eq:exp-ub}), 
\begin{align}
&\Var[\hat{\bxi}_{\ell, t}] \leq \Ex[\hat{\bxi}_{\ell, t}^2] \nonumber \\
	&\quad \leq (1+\eps)^2 \sum_{i =1}^n \sum_{j =1}^n \ind\{\alpha_i , \alpha_j  \in I_{\ell} \}  \cdot \sfK(x_i, y) \cdot \sfK(x_j, y) \cdot \Prx_{\bw_{\ell, t}}\left[ \beta + \bw_{\ell, t}^{1/s_2} \leq \min\{ \alpha_i, \alpha_j \} \right]. \label{eq:var-bound}
\end{align}
Suppose first that $\ell > 0$. Then if $\alpha_i, \alpha_j \in I_{\ell}$, then $(\max\{ \alpha_i, \alpha_j\} - \beta)^{s_2} \geq \sigma_{\ell}^{s_2}$, 
\begin{align}
 \Prx_{\bw_{\ell, t}}\left[ \beta + \bw_{\ell, t}^{1/s_2} \leq \min\{ \alpha_i, \alpha_j \} \right] &\leq \dfrac{((\min\{ \alpha_i, \alpha_j\} - \beta_j)^+)^{s_2}}{\sigma_{\ell+1}^{s_2} - \sigma_{\ell}^{s_2}} \cdot \dfrac{(( \max\{ \alpha_i, \alpha_j\} - \beta)^+)^{s_2}}{\sigma_{\ell}^{s_2}}\nonumber  \\
 		&\leq \left( \frac{((\alpha_i - \beta)^+)^{s_2}}{\sigma_{\ell+1}^{s_2} - \sigma_{\ell}^{s_2}} \cdot \frac{((\alpha_i - \beta)^+)^{s_2}}{\sigma_{\ell+1}^{s_2} -\sigma_{\ell}^{s_2}}\right) \cdot \left( \frac{\sigma_{\ell+1}^{s_2} - \sigma_{\ell}^{s_2}}{\sigma_{\ell}^{s_2}}\right) \nonumber  \\
		&\leq 2^{s_2} \cdot \Prx_{\bw_{\ell, t}}\left[ \beta + \bw_{\ell, t}^{1/s_2} \leq \alpha_i \right]   \Prx_{\bw_{\ell, t}}\left[ \beta + \bw_{\ell, t}^{1/s_2} \leq \alpha_j \right] .\label{eq:prob-bound}
\end{align}
Plugging (\ref{eq:prob-bound}) into (\ref{eq:var-bound}), we have that every $\ell \in \{1,\dots, k \}$ satisfies
\[ \Var\left[\hat{\bxi}_{\ell, t}\right] \leq (1+\eps)^2 \cdot 2^{s_2} \left( \Ex[\hat{\bxi}_{\ell, t}]\right)^2. \]
By the setting of $T$, we obtain the desired bound.
\end{proof}

\subsection{Proof of Lemma~\ref{lem:est-alpha} and Lemma~\ref{lem:est-beta}}\label{sec:est-alpha-proofs}

We briefly describe $\EstBeta$ (the case of $\EstAlpha$ is a symmetric argument, by replacing $\alpha$'s and $\beta$'s, as well as changing the signs). The intuition is the following: we initialize a data structure for $\AugmentedKDE$ with the kernel $\sfK$ in (\ref{eq:kernel}) and the parameter $s_2 = s-1$. This is done so that for every $i \in [n]$ and $j \in [m]$, 
\[ f_{ij}(\alpha_i, \beta_j) = \left(1 - \frac{1}{s} \right)^{s-1} \cdot \mu_i \left( (\alpha_i -\beta_j)^+ \right)^{s_2} \cdot \sfK(x_i, y_j). \]
Thus, we preprocess the data structure with the dataset $\{x_1,\dots, x_n \}$ and weights $\alpha \in \R^n$ and $\mu \in \R^n$. Then, we will iterate through each $j \in[m]$, and we query the data structure with $y_j$ and $\beta_j$ to obtain $\bxi_j$.
\ignore{In particular, for each $u \in U_{\ell}$, we first remove all points $x_i \in A_{u}$, then we query $\QueryA(v, x_j, \beta_j)$ for each $j \in B_u$, and then we add $x_i \in A_u$ with weight $\alpha_i$ back. The formal description appears in Figure~\ref{fig:lem-est-d}.

\begin{figure}
	\begin{framed}
		\noindent Algorithm $\EstAlpha(\Omega, u_{\ell}, \alpha,\beta, \eps)$.
		\begin{flushleft}
			\noindent {\bf Input:}  A ground set $\Omega = \{ x_1,\dots, x_n \}$, a function $u_{\ell} \colon [n] \to U_{\ell}$, two vectors $\alpha,\beta \in \R^n$ where $i \in A, j \in B$ satisfy $|\alpha_i + \beta_j| \in \{ 0 \} \cup [r_{\min} / M, r \cdot M]$ for $M = \poly(dn\Phi 2^s/\eps)$, and an error parameter $\eps > 0$. \\ 
			\noindent{\bf Output:} For every $u \in U_{\ell}$, an estimate $\hat{\bD}_u^{(B)} \in \R_{\geq 0}$ (the estimates for $\hat{\bD}_u^{(A)}$ are analogous). 
		\begin{enumerate}
				\item The kernel function $\sfK \colon \R^d \times \R^d \to \R_{\geq 0}$ of (\ref{eq:kernel}) with $\eps_0 = \eps / 2$, and we let $\tau, \delta$ be $1/\poly(nd\Phi2^s/\eps)$. Consider an arbitrary order of nodes in $U_{\ell}$ where $B_u$ is non-empty, and write $u_1,\dots, u_{m}$, where $m \leq n$.
				\item We initialize a data structure $v$ for $\AugmentedKDE(\sfK, s-1, \Phi, \eps/2, \tau, \delta)$, and for every $i \in A$, we execute $\UpdateAddA(v, x_i, \alpha_i)$.
				\item For each $k \in [m]$, we perform the following:
				\begin{itemize}
				\item For every $i \in A_{u_k}$, we execute $\UpdateRemA(v, x_i, \alpha_i)$.
				\item For every $j \in B_{u_k}$, we execute $\QueryA(v, y_j, \beta_j)$. If we write $\hat{\boldeta}_j$ for the output of $\QueryA(v, y_j, \beta_j)$, we let
				\[ \hat{\bD}_{u_k}^{(B)} = \left(1 - \frac{1}{s} \right)^{s-1}\sum_{j \in B_{u_k}} \hat{\boldeta}_j. \]
				\item For every $i \in A_{u_k}$, we execute $\UpdateAddA(v, x_i, \alpha_i)$.
				\end{itemize}
			\end{enumerate}
		\end{flushleft}
	\end{framed}
	\caption{Description for $\EstAlpha$ Algorithm.} \label{fig:lem-est-d}
\end{figure}}

\subsection{Proof of Lemma~\ref{lem:penalty}}\label{sec:est-penalty-proof}

The algorithm $\EstPenalty$ also uses Theorem~\ref{thm:aug-kde}. We initialize the kernel function $\sfK \colon \R^d \times \R^d \to \R_{\geq 0}$ as in (\ref{eq:kernel}), but with $s_2 = s$. Since $\rho = s / (s-1)$, we've set things up so
\begin{align*} 
(1-\eps_0) \left(f_{ij}(\alpha_i, \beta_j) \cdot \| x_i - y_j \|_{2}\right)^{\rho} &\leq \left(1 - \frac{1}{s}\right)^{s_2} \cdot ((\alpha_i - \beta_j)^+)^{s} \cdot \sfK(x_i, y_j) \\
	&\leq \left( f_{ij}(\alpha_i,\beta_j) \cdot \|x_i - y_j \|_2\right)^{\rho},
\end{align*}
and thus $\EstPenalty(\alpha,\beta, \eps, \delta)$ needs to approximate
\[ \sum_{i=1}^n \sum_{j=1}^m \mu_i \nu_j \cdot ((\alpha_i - \beta_j)^+)^{s_2} \cdot \sfK(x_i, y_j). \]
This is a simple application of Theorem~\ref{thm:aug-kde}. We preprocess a data structure $v$ for $\AugmentedKDE(\sfK, s, \Phi, \eps/2, \delta)$ with the dataset $\{x_1,\dots, x_n\}$ and weights $\alpha$, and we query it for each $y_1,\dots, y_m$ with the weights $\beta \in \R^m$. If each estimate is $\hat{\bxi}_j$, the desired output estimate is given by $\sum_{j=1}^m \hat{\bxi}_j$.


\end{document}